\title{Reputational Bargaining and Inefficient Technology Adoption\footnote{We thank Dilip Abreu, Sandeep Baliga, Francesc Dilm\'{e}, Rohan Dutta,
Mehmet Ekmekci, Jack Fanning, Faruk Gul, Marina Halac, Richard Holden, Ruitian Lang, Lucas Maestri, Chiara Margaria, Stephen Morris, Juan Ortner, Larry Samuelson, 
Karthik Sastry, Bruno Strulovici, Chris Udry, Alex Wolitzky, and Hanzhe Zhang for helpful comments. We thank the NSF Grants SES-1947021 and SES-2337566 and the Cowles Foundation for financial support.}}
\author{Harry Pei\footnote{Department of Economics, Northwestern University. Email: harrydp@northwestern.edu} \and
Maren Vairo\footnote{Department of Economics, Princeton University. Email: mvairo@princeton.edu}}
\date{\today}

\documentclass[12pt]{article}
\usepackage{amsmath}
\usepackage{float}
\usepackage{graphicx}
\usepackage{amsfonts}
\usepackage{amsthm}
\usepackage{amssymb}
\usepackage{setspace}
\usepackage{tikz}
\usetikzlibrary{patterns}
\usepackage{sgame}
\usepackage{color}
\usepackage{hyperref}
\usepackage{subfig}
\usepackage{graphicx}
\DeclareMathOperator*{\argmax}{arg\,max~}
\DeclareMathOperator*{\argmin}{arg\,min~}
\usepackage[top=1.25in, bottom=1.25in, left=1.25in, right=1.25in]{geometry}
\begin{document}
\newtheorem{Proposition}{\hskip\parindent\bf{Proposition}}
\newtheorem{Theorem}{\hskip\parindent\bf{Theorem}}
\newtheorem{Lemma}{\hskip\parindent\bf{Lemma}}
\newtheorem{Corollary}{\hskip\parindent\bf{Corollary}}
\newtheorem*{Definition}{\hskip\parindent\bf{Definition}}
\newtheorem{Assumption}{\hskip\parindent\bf{Assumption}}
\newtheorem*{Assumption0}{\hskip\parindent\bf{Assumption}}
\newtheorem{Claim}{\hskip\parindent\bf{Claim}}
\newtheorem*{Refinement}{\hskip\parindent\bf{Equilibrium Refinement}}

\maketitle
\numberwithin{equation}{section}

\noindent \textbf{Abstract:} A seller decides whether to adopt a new technology that can lower his cost of production and then bargains with a buyer over the price of an object. Both players can build reputations for being obstinate in the bargaining process by offering the same price over time. We show that even when the buyer \textit{cannot} observe the seller's adoption decision, players' incentives to build reputations can lead to inefficient under-adoption. We also show that under-adoption occurs in equilibrium \textit{if and only if} there are significant delays in reaching agreements, and that these inefficiencies can arise \textit{if and only if} the social benefit from adoption is large enough. As a result, an increase in the social benefit from adoption may lower the probability of adoption and welfare.\\


%\noindent \textbf{JEL Codes:} C73, D82, D83

\begin{spacing}{1.5}
\section{Introduction}\label{sec1}
Suppose a supplier needs to decide whether to adopt a new technology that can lower his cost of production. Even when the gains from adoption outweigh the costs, the supplier might be reluctant to adopt due to the concern that after his investment becomes sunk cost, his clients will offer low prices and expropriate the gains from adoption.
This is the well-known \textit{hold-up problem}, which is a fundamental determinant of people's incentives to make relationship-specific investments, firms' incentives to adopt new technologies, as well as the boundaries of firms and organizations.

The severity of the hold-up problem depends on the bargaining process that determines the terms of trade as well as players' information about others' investment decisions. For example, Grossman and Hart (1986) assume that bargaining is efficient and that investments are publicly observed. They show that investments are inefficient unless the player who makes the investment decision has all the bargaining power. Gul (2001) shows that investments are approximately efficient even when the investing player cannot make any offer and hence has no bargaining power, as long as his opponents can \textit{frequently revise their offers} and \textit{cannot} observe how much he has invested.

We revisit the hold-up problem by incorporating an important concern in practice, that players 
may have incentives to build reputations for being obstinate in the bargaining process and as a result, might be
reluctant to revise their offers. In contrast to Gul (2001), we show that even when investments are unobservable,
reputational incentives can lead to \textit{inefficient under-investment}.\footnote{Inefficient investment also occurs in reputational bargaining games where players' investments are observable.} 
Furthermore, under-investment arises \textit{if and only if} there are delays in reaching agreement, and these inefficiencies can occur \textit{if and only if} the social benefit from investment is large enough.

%\footnote{Our analysis also applies to the following situations. First, when the seller decides whether to divest, i.e., receiving a benefit for becoming less cost-efficient. Second, when the buyer (e.g., a downstream firm) decides whether to increase her value (e.g., expands its customer base) by paying some costs (e.g., the cost of advertising).}

We augment the reputational bargaining model of Abreu and Gul (2000) with a technology adoption stage \textit{before} the bargaining stage. A buyer and a seller bargain over the price of an object. The buyer's value is commonly known. The seller's production cost is his private information, which depends on his choice of production technology at the adoption stage \textit{before} bargaining starts. In our baseline model, the seller either uses a \textit{default technology}, or adopts a \textit{new technology} that has a lower production cost compared to the default one but requires a positive adoption cost.

In the bargaining stage, the buyer offers a price.\footnote{The uninformed player making the offer first is a standard assumption in the reputational bargaining literature which is also assumed in Abreu, Pearce and Stacchetti (2015) and Fanning (2024). If the informed player makes the first offer, then their signaling incentives make the game intractable to analyze, which we leave for future research.}
The seller either accepts the buyer's offer, or demands a higher price after which players engage in a continuous-time war-of-attrition. With positive probability, each player is one of \textit{a rich set of commitment types} who offers an exogenous price and never concedes. With complementary probability, they are rational and decide what to offer and when to concede in order to maximize their discounted payoff. %As in Abreu and Gul (2000), our analysis focuses on the case in which commitment types occur with low probability. 

Theorem \ref{Theorem1} characterizes equilibria of a reputational bargaining game where the distribution of the seller's production cost is \textit{exogenous}. It shows that, fixing the distribution over costs and as the probability of commitment types vanishes, inefficient delays arise in equilibrium \textit{if and only if} (i) the difference between the two production costs is large enough, and (ii) the seller has a low production cost with probability above some cutoff. Our inefficient bargaining result stands in contrast to the results in Kambe (1999), Abreu and Gul (2000), and Abreu, Pearce and Stacchetti (2015, or APS), which show that bargaining is efficient when players have no private information about their payoffs, or when they only have private information about their discount rate.

The bargaining inefficiencies stem from the buyer's incentive to \textit{screen} the seller, that is, to induce sellers with different costs to demand different prices. Screening is feasible when the buyer faces uncertainty about the seller's cost but not when she faces uncertainty only about the seller's discount rate. To see why, suppose the buyer makes a \textit{screening offer} that is \textit{between} the two costs. The low-cost seller gets a positive payoff from conceding, so he faces the trade-off identified in Abreu and Gul (2000) that demanding more surplus lowers his speed of building reputations. 
However, this trade-off is irrelevant for the high-cost type since his payoff from conceding is negative. As a result, 
the high-cost type prefers to demand more surplus when doing so can signal his cost. We show that
as long as a certain set of commitment types occur with positive probability, following any screening offer (i) the high-cost type will demand the entire surplus  and trade with delay and 
(ii) the low-cost type will trade immediately at a lower price since delays are more costly for him.

When is screening profitable for the buyer? As in APS, the buyer can offer the high-cost type's Rubinstein bargaining price (i.e., make a \textit{pooling offer}), which will induce both types of the seller to trade immediately. She can also make a \textit{screening offer} between the two production costs, after which she will lose all her surplus when she faces the high-cost type. Hence, the buyer prefers to make the screening offer only if (i) she can pay a lower price to the low-cost type and (ii) the low-cost type occurs with high enough probability. The former is true \textit{if and only if} the difference between the two costs is large enough. This is because when the cost difference is small, 
all screening offers leave too little surplus to the low-cost type. The low-cost type can then build reputation faster than the buyer even when he demands prices that are greater than the pooling offer.

Our main result, Theorem \ref{Theorem2}, shows that when the adoption decision is \textit{endogenous}, 
the seller's adoption decision can be inefficient \textit{if and only if} the social benefit from adoption (defined as the difference between the high and the low production cost) is large enough.\footnote{The equilibrium level of adoption depends on the adoption cost as well. We show that there exists a cutoff such that inefficient adoption can arise \textit{only if} the social benefit from adoption is above the cutoff. Moreover, \textit{if} the social benefit from adoption is large enough \textit{and} the adoption cost belongs to an open set of ``intermediate values,'' then the unique equilibrium has inefficient adoption.} Otherwise, his adoption decision is \textit{approximately efficient} regardless of the adoption cost. This result implies that an increase in the social benefit from adoption may, paradoxically, \textit{lower} the probability of adoption. 

To understand Theorem \ref{Theorem2}, we start from an observation that the seller can capture all the surplus from adoption when the buyer makes the pooling offer but cannot do so when the buyer makes the screening offer. This together with an implication of Theorem \ref{Theorem1}, that the gap between the production costs being large is necessary for the buyer to prefer making screening offers, seems to explain why inefficient adoption can arise only when the social benefit from adoption is large.

However, this explanation is misleading since Theorem \ref{Theorem1} \textit{cannot} be applied when the distribution of production cost is endogenous: It only applies once we \textit{fix} the distribution of production cost \textit{as the probability of commitment types vanishes}. But in the game with endogenous technology adoption, the distribution of production cost depends on the probability of commitment types.

Instead, to provide a valid explanation for Theorem \ref{Theorem2}, we use the observation that \textit{fixing any probability of commitment types}, the buyer prefers low-price screening offers when the probability that the seller's cost is high goes to zero. In contrast, Theorem \ref{Theorem1} says that for any \textit{fixed distribution of production costs}, as the probability of commitment types vanishes, the buyer prefers to make a high-price pooling offer when the difference between the high and low production costs is small.

The probability that the seller has a high production cost being small relative to the probability of the commitment types
is indeed what may happen in the game with endogenous adoption. When the social benefit from adoption is small but is greater than the seller's cost of adoption, the equilibrium adoption probability converges to one as the probability of commitment types vanishes. 
When the social benefit from adoption is large (i)
the seller may adopt with probability bounded below one despite adoption is socially efficient, and (ii)
the buyer makes the low-price screening offer with non-vanishing probability. Delays occur when the buyer makes the screening offer and the seller has a high production cost. They occur with non-vanishing probability, and hence have significant welfare consequences, only when the social benefit from adoption is large enough.

 We extend the insight of Theorem \ref{Theorem2} to a setting with a rich but finite space of seller-cost types in Theorem \ref{Theorem4}. There, we show that if the gap between the highest and lowest seller-cost types is large enough, then there exists an open set of investment costs under which inefficient adoption and bargaining delay can arise in equilibrium, echoing the finding in Theorem \ref{Theorem2}.

Our results suggest an explanation for the under-adoption of cost-saving technologies, which is widely documented in agriculture and manufacturing. Overall, our theory fits when (i) the producers \textit{know} that the technology can effectively lower cost,\footnote{Wolitzky (2018) provides an alternative explanation based on social learning instead of the hold-up problem. His theory fits when the producers do \textit{not} know the effectiveness of the new technology.} (ii) it is hard for buyers to observe the producers' adoption decisions, but (iii) the producers are reluctant to adopt due to the fear of being held up.

% One example that fits is the under-adoption of Bt cotton. It is well-known among farmers that Bt cotton can reduce insecticide applications which can lower the cost per unit yield  (Qaim and de Janvry 2003). Although the \textit{crop} of Bt cotton is more resistant to pests compared to that of traditional cotton, it is hard for buyers to distinguish between the two since (i) the crops have similar appearances and (ii) both crops lead to the same final product, i.e., cotton. Some farmers in Pakistan sampled by Ali and Abdulai (2010) are major landholding households, who seem to have bargaining power. 
% However, the adoption rate is low even among those households, e.g., it is only 62\% in the Punjab province of Pakistan (Ali and Abdulai 2010). 
% Although there are other explanations, such as the lack of access to credit markets, farmers' concerns about the hold-up problem also seem to be relevant given that (i) the adopted farmers are more likely to be members of organizations that have more bargaining power, and (ii) the farmers' share of surplus is much lower than that in countries that have higher adoption rates (Falck-Zepeda, et al. 2000). 

In the remainder of this section, we discuss our contributions to the related literature. Section \ref{sec2} sets up the baseline model in which the seller chooses between two production technologies. The main results are stated in Section \ref{sec3} and are shown in the appendices. Section \ref{sec5} extends our results to settings where the seller chooses between three or more technologies. 
Section \ref{sec6} concludes. 

\paragraph{Bargaining \& Reputations:} We take a first step to analyze reputational bargaining games when the distribution of players' preferences is \textit{endogenous}. This aspect is novel relative to the existing works in this literature, which assume that the distribution of preferences is exogenous.

Compared to Abreu and Gul (2000), we introduce \textit{heterogeneity} in players' costs and show that it enables the uninformed player to \textit{screen} the informed player via unattractive offers, leading to inefficient delays.\footnote{Ekmekci and Zhang (2024) study reputational bargaining with interdependent values but with only one rational type for each player. In contrast, we study a private value model where the seller has multiple rational types.} This stands in contrast to the model analyzed by APS where 
the only private information is about a player's \textit{discount rate} and whether players are committed, and the limiting equilibrium outcome is efficient when commitment types' demands are constant over time.\footnote{APS also consider the case in which some commitment types' demands change over time. Since our motivation is to revisit the hold-up problem when players are unwilling to revise their offers, we assume that all commitment types demand the same price over time, which is also assumed in Abreu and Gul (2000) and Fanning (2024).}

A contemporaneous work of Fanning (2024) focuses on the interaction between private information about values (or equivalently, costs) and outside options, but not on endogenous investments. 
He shows that bargaining is efficient when 
(i) the informed player's value is drawn from a rich set
and
(ii) no rational type is indifferent between accepting any commitment type's offer and taking the outside option. 
 Our inefficient bargaining result requires the existence of two adjacent types with sufficiently different production costs,\footnote{Ortner (2017) shows that when a seller's cost may decrease over time, the outcome is efficient if and only if the buyers' values are drawn from a rich set. Although our analysis reaches a similar conclusion, the inefficiencies in our model are driven by players' incentives to build reputations, which is not the case for Ortner (2017)'s result.} which violates his richness requirement on the set of values. Our modeling assumption fits when the heterogeneity in production cost is driven by the \textit{differences in production technologies}, as adoption decisions are usually \textit{indivisible} and adopting an innovative technology can significantly lower production cost. 
His second requirement violates our richness assumption on the set of commitment types, which we will explain in Section \ref{sub2.1}.

Our model has private values and there is a \textit{gap} between the seller's cost and the buyer's value. In contrast to Gul, Sonnenschein and Wilson (1986) who show that there is no delay as the bargaining friction vanishes, we show that players' reputation concerns can cause delays. Our result stands in contrast to the inefficiency results that are driven by interdependent values (Deneckere and Liang 2006, Baliga and Sj\"{o}str\"{o}m 2023), 
no gap between players' values (Ausubel and Deneckere 1989),
costly concessions (Dutta 2023), risk aversion (Dilm\'{e} and Garrett 2022),
and the arrival of new traders (Fuchs and Skrzypacz 2010).

\paragraph{Hold-Up Problem:} Compared to Grossman and Hart (1986) and Gul (2001), 
we incorporate a practical concern that players might be reluctant to revise their offers due to their incentives to build reputations for being obstinate.\footnote{Instead of assuming that the investing player cannot make any offer and has no bargaining power, we diverge from Gul (2001) by allowing both players to make offers, which seems to provide the investing player stronger investment incentives. However, in contrast to Gul (2001)'s efficient investment result, we show that investment can be inefficiently low. This comparison  highlights the role of reputation concerns in driving these inefficiencies.} 
We show that the hold-up problem re-emerges even when the investment decision is \textit{unobservable}. Our result implies that the \textit{absolute magnitude} of the social benefit from investment has a significant effect on players' investment incentives.
This is complementary to the existing theories in which a player's investment incentive depends only on the \textit{ratio} between the social benefit from investment and the cost of investment.\footnote{In Che and S\'{a}kovics (2004), investment incentives depend on the difference between the social surplus and the cost of investment, whereas investment incentives depend both on this ratio and on the \textit{absolute magnitude} of social surplus in our model.}

Our conclusion also applies to other forms of \textit{selfish} investments, such as when 
the seller decides whether to \textit{divest} and become less cost-efficient. When investments are \textit{cooperative} 
such as the seller's investment increases the buyer's value (Che and Hausch 1999), the seller has no incentive to invest when his investment is unobservable, but has an incentive to invest when it is observable. This stands in contrast to our model in which non-observability leads to stronger investment incentives.

\section{The Baseline Model}\label{sec2}
 Time is continuous, indexed by $t \in [0,+\infty]$. 
A buyer (she) and a seller (he) bargain over the price of an object. The buyer's value for the object is common knowledge, which we normalize to $1$.

Time $0$ consists of two stages. 
In the first stage, the seller decides whether to adopt a new technology at an \textit{adoption cost} $c>0$.
This adoption decision determines his cost of producing the object, which the buyer \textit{cannot} observe.
If he adopts the new technology, then his production cost is $\theta_1$. If he uses the default technology, then his production cost is $\theta_2$, with $0<\theta_1<\theta_2<1$. We extend our results to settings with three or more production technologies in Section \ref{sec5}.

In the second stage, the buyer offers a price $p_b \in [0,1]$.  The seller either accepts the offer and sells at price $p_b$, 
or rejects the offer and makes a counteroffer $p_s \in (p_b,1]$, after which players engage in a continuous-time war-of-attrition. If a player concedes, then players trade at the price offered by their opponent. If both players concede at the same time, then they trade at price $\frac{p_b+p_s}{2}$.

Players share the same discount rate $r>0$.
If trade happens at time $\tau \in [0,+\infty]$ and price $p \in [0,1]$, then the buyer's payoff is $e^{-r \tau} (1-p)$
and the seller's payoff is $e^{-r \tau} (p-\theta)-\widetilde{c}$, where $\theta$ stands for the seller's production cost and $\widetilde{c} \in \{0,c\}$ stands for his adoption decision. 
If players never trade, then $\tau=+\infty$, in which case
the buyer's payoff is $0$ and the seller's payoff is $-\widetilde{c}$.

Each player is rational with probability $1-\varepsilon$ and is one of the commitment types with probability $\varepsilon>0$. 
Each buyer-commitment-type is characterized by $p_b \in \mathbf{P}_b \subset [0,1]$, who offers $p_b$ and never accepts any price that is strictly greater than $p_b$. Each seller-commitment-type is characterized by $p_s \in \mathbf{P}_s \subset [0,1]$, who offers $p_s$ and never accepts any price that is strictly lower than $p_s$. Let $\mu_b \in \Delta (\mathbf{P}_b)$ and $\mu_s \in \Delta (\mathbf{P}_s)$ denote the distributions of players' commitment types conditional on them being committed, both of which have full support.  
We do not specify the commitment types' production costs 
since they are irrelevant for the rational types' incentives and behaviors.  

We assume that the buyer's and the seller's types are independently distributed and that the sets of commitment types $\mathbf{P}_b$ and $\mathbf{P}_s$ are \textit{compact} and \textit{countable}. 
We also assume that $\mathbf{P}_s$ is \textit{rich} in the sense that $1 \in \mathbf{P}_s$ and $\sup\mathbf{P}_s\setminus \{1\}=1$. That is to say, the seller \textit{can} build a reputation for demanding the entire surplus and also for demanding something strictly less than but close to the entire surplus.\footnote{Our assumption is satisfied, for example, when there exists $\nu \in (0,1)$ such that the set of seller-commitment-types \textit{contains} $1$ and $p^j \equiv 1-(1-\nu)^j$ for every $j \in \mathbb{N}$. Our assumption allows other commitment types to exist as well.}
We will discuss this richness assumption in Section \ref{sub2.1}. For future reference, 
let
\begin{equation}\label{richness}
\nu \equiv \inf \Big\{ \overline{\nu}>0 \Big|
\textrm{ for every $p \in [0,1]$, $(p-\overline{\nu},p+\overline{\nu}) \cap \mathbf{P}_s \neq \emptyset$ and $(p-\overline{\nu},p+\overline{\nu}) \cap \mathbf{P}_b \neq \emptyset$ } 
\Big\}.
\end{equation}
If $\nu$ is close to $0$, then for each price $p \in[0,1]$ and each player, there exists a commitment type for that player whose demand is close to $p$.

We adopt two modeling conventions which are common in the literature and are also adopted by APS and Fanning (2024). First, regardless of the seller's offer $p_s$, if he accepts the buyer's offer at time $0$ or concedes to the buyer at time $0$ before the buyer concedes, then we view his offer as $p_b$.
Second, if the buyer's strategy is to offer $p_b\notin\mathbf{P}_b$ and to accept any counteroffer from the seller, then we relabel her strategy as offering $\min\mathbf{P}_b$ and then accepting any counteroffer by the seller.

%\footnote{As shown in APS, under this convention, it is without loss of generality to focus on equilibria where the support of the buyer's offer is a subset of $\mathbf{P}_b$ and the support of the seller's offer following $p_b$ is a subset of $\mathbf{P}_s\cup\{p_b\}$.} 

The public history consists of players' offers and whether any player has conceded. The buyer's private history consists of the public history and whether she is committed. The rational-type buyer's strategy consists of her offer $\sigma_b\in \Delta[0,1]$ and a mapping from players' offers to her concession time $\widetilde{\tau}_b:[0,1]^2\to \Delta[0,+\infty]$. The seller's private history consists of the public history, whether he is committed, and his adoption decision. The rational-type seller's strategy consists of his adoption decision, or equivalently, the distribution of his production cost $\pi \in \Delta \{\theta_1,\theta_2\}$, a mapping from his production cost and the buyer's offer to his offer $\sigma_s: \{\theta_1,\theta_2\} \times [0,1]\to \Delta[0,1]$, and a mapping from his production cost and players' offers to his concession time $\widetilde{\tau}_s: \{\theta_1,\theta_2\} \times [0,1]^2\to \Delta[0,+\infty]$. The solution concept is Perfect Bayesian equilibrium, or \textit{equilibrium} for short.

For future reference, let $\sigma_b(p_b)$ and $\sigma_s(p_s|\theta,p_b)$ denote the probabilities with which the buyer's strategy $\sigma_b$ and type-$\theta$ seller's strategy $\sigma_s(\theta)(p_b)$ assign to offers $p_b$ and $p_s$, respectively.
%These probabilities are well-defined since $\mathbf{P}_b$ and $\mathbf{P}_s$ are countable sets and the distributions over offers are supported on these sets under our conventions. 
Let $\tau_b \in \Delta [0,+\infty]$ and $\tau_s \in \Delta [0,+\infty]$ denote the unconditional distributions of the buyer's and the seller's concession times, which can be computed from the equilibrium strategies and type distributions.

\subsection{Discussions on the Modeling Assumptions}\label{sub2.1}
We analyze a reputational bargaining model since our motivation is to revisit the hold-up problem when players may not want to change their bargaining postures due to their reputation concerns. Compared to bargaining models with incomplete information but \textit{without} commitment types such as Gul, Sonnenschein and Wilson (1986) and Gul (2001),  reputational bargaining models lead to sharp predictions on players' welfare even when both players have some bargaining power.
This sounds more realistic relative to the restriction that one of the players has no bargaining power.\footnote{It is well-known that bargaining models where the informed player can make offers are not tractable to analyze. As a result, most of the existing works focus on the case where the uninformed player makes all the offers (e.g., Gul, Sonnenschein and Wilson 1986, Gul 2001). An exception is Gerardi, H\"{o}rner and Maestri (2014) that characterizes the set of equilibrium payoffs when the informed player makes \textit{all} the offers. }
% We are unaware of any paper that analyzes models \textit{without} any commitment type where both the informed and uninformed player can make offers.

We assume that the \textit{uninformed} player (i.e., the buyer) makes an offer before the \textit{informed} player (i.e., the seller) does. This assumption is also made in APS and Fanning (2024). The uninformed buyer making the first offer seems plausible when a firm procures inputs from its upstream supplier (e.g., farmers), in which case the firm usually quotes a price before the negotiations. If the informed seller makes the first offer, then his incentive to signal his production cost complicates the analysis. We leave this case for future research.

We assume that the set of commitment types is \textit{rich} in the sense that there exists a seller-commitment-type who demands the entire surplus  as well as a sequence of seller-commitment-types whose demands are strictly less than but close to the entire surplus. Our richness requirement on the set of commitment types is violated in Fanning (2024) who assumes that \textit{no rational type is indifferent between accepting any offer made by any commitment type and taking the outside option}. Under his assumption, there exists no commitment type who demands the entire surplus.\footnote{Although in Abreu and Gul (2000)'s model, all commitment types' demands are strictly less than the entire surplus, their main result that players will trade immediately at the Rubinstein bargaining price \textit{extends} to the case where the set of commitment types satisfies our richness assumption. This is because in any equilibrium of their model, the rational type has no incentive to imitate the commitment type who demands the entire surplus.}

The motivation for our richness assumption is that a player should be able to build a reputation for being obstinate as long as they demand the same price over time, \textit{regardless of what their demand is}.\footnote{We assume that the number of commitment types is countable in order to circumvent measurability issues. This assumption is made in most of the existing reputational bargaining models, which include Abreu and Gul (2000).} 
We show that inefficient equilibria \textit{exist} as long as there \textit{exists} a commitment type who demands the entire surplus, and that 
\textit{all} equilibria are inefficient when there also \textit{exists} a sequence of commitment types whose demands are less than but are arbitrarily close to the entire surplus.
Our findings are robust to the inclusion of any additional stationary commitment type.
In that sense, they are in the spirit of the reputation results in Fudenberg and Levine (1989) and Abreu and Gul (2000) that these conclusions apply \textit{as long as a certain set of commitment types occur with strictly positive probability}, even when there might be other commitment types as well.\footnote{For instance, Fudenberg and Levine (1989) show that the patient player can secure their Stackelberg payoff as long as there \textit{exists} a commitment type who takes the Stackelberg action. Abreu and Gul (2000) show that players obtain their Rubinstein bargaining payoffs as long as there \textit{exist} commitment types who demand those payoffs.}

One caveat of this approach, however, is that the convergence of discrete-time alternating-offers bargaining games to the continuous-time war of attrition, as offers become very frequent (Abreu and Gul, 2000), fails under our richness assumptions on the set of commitment types. In particular, as Fanning (2024) points out, if time is discrete, no rational player would benefit from mimicking a commitment type that demands the entire surplus. Consequently, our results cannot be directly applied to characterize the equilibria of the discrete-time game.

Finally, we restrict attention to \textit{stationary commitment types} by requiring every commitment type to demand the same price over time. This restriction is used as a common benchmark in the reputational bargaining literature and is also assumed in Kambe (1999), Abreu and Gul (2000), Ekmekci and Zhang (2022), and Fanning (2024). The analysis of reputational bargaining with, even exogenous, private information about rational players' values in the presence of non-stationary commitment types remains an interesting open question.\footnote{Abreu and Pearce (2007), Wolitzky (2012), APS, and Fanning (2016, 2018) consider the role of non-stationary commitment types in various reputational bargaining settings, none of which include private information about rational players' values.}
% It is motivated by a practical concern that once a player changes their demand, it might be hard for them to convince others that they are obstinate. 

%in Fudenberg and Levine (1989) and Abreu and Gul (2000) in the sense that \textit{ideally},  (e.g., commitment types that play the Stackelberg action in Fudenberg and Levine and  commitment types that demand the Rubinstein bargaining price in Abreu and Gul), even when there might be other commitment types as well.

\section{Analysis \& Results}\label{sec3}
Section \ref{sub2.2} analyzes a benchmark where the buyer can observe the seller's adoption decision. Section \ref{sub3.1} analyzes a reputational bargaining game where the seller's cost is drawn from an \textit{exogenous} distribution and is his private information. Section \ref{sub3.2} analyzes  reputational bargaining with \textit{endogenous} technology adoption. Our main result, Theorem \ref{Theorem2}, shows that costly delays and inefficient technology adoption can arise in equilibrium \textit{if and only if} there are large social gains from adoption. We also explain the subtleties when analyzing models with endogenous cost distributions. 

\subsection{Benchmark: Adoption Decision is Observable}\label{sub2.2}
Suppose the buyer \textit{can} observe the seller's adoption decision, that is, the buyer knows $\theta$. Proposition 3 in Abreu and Gul (2000) implies that as the probability of commitment types $\varepsilon$ and $\nu$ defined in (\ref{richness}) go to $0$, players will trade with no delay at a price close to $p_{\theta} \equiv \frac{1+\theta}{2}$. We call $p_{\theta}$ type-$\theta$ seller's \textit{Rubinstein bargaining price} since it is 
the equilibrium price in the bargaining game of Rubinstein (1982) between a buyer with value $1$ and a seller with cost $\theta$ when players make offers frequently.

The intuition is that the buyer can secure payoff $1-p_{\theta}$ by offering $p_{\theta}$ and the seller can secure payoff $p_{\theta}-\theta$ by demanding $p_{\theta}$. Their guaranteed payoffs are their equilibrium payoffs since the sum of these payoffs equals the social surplus from trade $1-\theta$. 
The seller's gain from adoption is
$(p_{\theta_1}-\theta_1) - (p_{\theta_2}-\theta_2)= \frac{\theta_2-\theta_1}{2}$, which implies that he will adopt 
only when
$c \leq \frac{\theta_2-\theta_1}{2}$.

Since it is socially efficient to adopt the new technology as long as $c<\theta_2-\theta_1$, the equilibrium adoption decision is \textit{inefficient} when $c \in \big(\frac{\theta_2-\theta_1}{2},\theta_2-\theta_1 \big)$. In summary, when the seller's adoption decision is observable, there is almost no delay in reaching an agreement but the adoption decision is socially inefficient under an open set of adoption costs.

\subsection{Reputational Bargaining with Exogenous Production Cost}\label{sub3.1}
This section analyzes a reputational bargaining game when the rational seller's production cost $\theta$ is drawn from an \textit{exogenous} full support distribution $\pi \in \Delta \{\theta_1,\theta_2\}$. We refer to the rational seller with production cost $\theta$ as \textit{type $\theta$}. We also use type $\theta_1$ (type $\theta_2$) and the \textit{low type} (\textit{high type}) interchangeably. 
We start from defining several strategies. Let $\underline{\sigma}_{b}$ denote the buyer's strategy of offering $\min \{p_{\theta_1},\theta_2\}$.
Let $\overline{\sigma}_{b}$ denote the buyer's strategy of offering $p_{\theta_2}$.  
By definition, $p_{\theta_2} > \min \{p_{\theta_1},\theta_2\}$.
Let $\sigma_s^* (\cdot) \equiv \big\{ \sigma_{s,\theta}^* (\cdot) \big\}_{\theta \in \Theta}$, where $\sigma_{s,\theta}^* (p_b)$ assigns probability $1$ to $p_{s,\theta}^* (p_b)$ with
\begin{align}\label{equilibriumoffers}
    p_{s,\theta}^*(p_b) \equiv 
    \begin{cases} 
    1, \quad & \text{if  } p_b\leq \theta,\\
    \max\big\{p_b,1+\theta_1-p_b\big\}, \quad & \text{if  } p_b\in (\theta_1,\theta_2] \text{  and  } \theta=\theta_1, \\ 
    \max\big\{p_b,1+\theta_2-p_b\big\}, \quad & \text{if  } p_b>\theta_2.
    \end{cases}
\end{align}
Later on, we will show in Theorem \ref{Theorem1} that $\sigma_{s,\theta}^* (\cdot)$ is type-$\theta$ seller's equilibrium counteroffer.

In order to understand the expression for $\sigma_{s,\theta}^*$, we start from explaining the intuition behind $\max\big\{p_b,1+\theta_i-p_b\big\}$. Recall that in a reputational bargaining game where it is common knowledge that $\theta=\theta_i$, for any pair of offers $p_b$ and $p_s$ with
$\theta_i<p_b<p_s<1$, the seller will concede at rate
\begin{equation}\label{sellerconcession}
\lambda_s \equiv \frac{r (1-p_s)}{p_s-p_b}, 
\end{equation}
and the buyer will concede at rate
\begin{equation}\label{buyerconcession}
\lambda_b^i \equiv \frac{r(p_b-\theta_i)}{p_s-p_b}. 
\end{equation} 
These are the rates that make the other player indifferent between conceding and not conceding. 

Proposition 3 in 
Abreu and Gul (2000) implies that as the probability of commitment types $\varepsilon$ vanishes, the player with a lower concession rate will concede at time $0$ with probability close to $1$. 
If $\theta_i$ is common knowledge, then
 (\ref{sellerconcession}) and (\ref{buyerconcession}) imply that players will concede at the same rate when the seller offers $1+\theta_i-p_b$. This implies that
the seller can secure a price of approximately 
$\max\big\{p_b,1+\theta_i-p_b\big\}$ 
\textit{either} by accepting the buyer's offer \textit{or} by offering something slightly below $1+\theta_i-p_b$ and inducing the buyer to concede at time $0$ with probability close to $1$. 

Let
\begin{align}\label{cutoff}
     \pi^*\equiv \min \Big\{1, \frac{p_{\theta_{2}}-\theta_{2}}{\min\{p_{\theta_1},\theta_{2}\}-\theta_1} \Big\},
 \end{align}
which by definition is strictly positive. One can verify that $\pi^*<1$ \textit{if and only if}
\begin{equation}\label{costdifference}
\theta_2-\theta_1>\frac{1-\theta_2}{2},
\end{equation}
that is, the difference between $\theta_1$ and $\theta_2$ is large relative to the surplus generated by the high-cost type. Theorem \ref{Theorem1} shows that for generic $\pi$, all equilibria converge to the same limit point when the sets of commitment types satisfy our richness assumption and the probability of commitment types vanishes.\footnote{Throughout the paper, we measure the distance between two distributions (e.g., two mixed actions) using the Prokhorov metric, defined in Billingsley (2013a). Intuitively, two distributions $\mu$ and $\mu'$ are close if for every Borel set $A$, the value of $\mu(A)$ is close to that of $\mu'(A')$ for some small neighborhood $A'$ of $A$.} It also characterizes the welfare properties of the unique limiting equilibrium.
\begin{Theorem}\label{Theorem1}
For every $\pi \in \Delta \{\theta_1,\theta_2\}$, there exists at least one equilibrium. 
Suppose in addition that $\pi$ satisfies $\pi (\theta_1) \notin \{0,\pi^*,1\}$. For every $\eta>0$, there exists $\bar{\nu}>0$ such that when $\nu<\bar{\nu}$, there exists $\bar{\varepsilon}_{\nu} >0$ such that for every $\varepsilon \in (0,\bar{\varepsilon}_{\nu})$ and every equilibrium $(\sigma_s,\sigma_b,\tau_s,\tau_b)$ under $(\varepsilon, \nu)$:
\begin{enumerate}
\item If $\pi(\theta_1)<\pi^*$, then $\sigma_b$ is $\eta$-close to $\overline{\sigma}_{b}$, and $\sigma_s(p_b)$ is $\eta$-close to $\sigma_{s}^*(p_b)$ for every $p_b$ such that $\sigma_b(p_b)>\eta$. The expected welfare loss from delay is less than $\eta$ conditional on every $\theta \in \Theta$.
\item If $\pi(\theta_1)>\pi^*$, then $\sigma_b$ is $\eta$-close to $\underline{\sigma}_{b}$, and $\sigma_s(p_b)$ is $\eta$-close to $\sigma_{s}^*(p_b)$ for every $p_b$ such that $\sigma_b(p_b)>\eta$. Conditional on $\theta=\theta_1$, the expected welfare loss from delay is less than $\eta$. Conditional on $\theta=\theta_2$, the buyer's equilibrium payoff is $0$ and the expected welfare loss from delay is $\eta$-close to
\begin{equation}\label{expecteddelay}
(1-\theta_2) \Bigg\{ 1-\frac{1-\min\{p_{\theta_1},\theta_2\}}{1-\theta_1}\Bigg\}. 
\end{equation}
\end{enumerate}
\end{Theorem}
The proof is in Appendix \ref{sub4.1} and an intuitive explanation is provided later in this section. 
According to Theorem \ref{Theorem1}, under generic $\pi$, there exists a unique limiting equilibrium. 
The qualitative features of the unique limiting equilibrium depend on (i) the difference $\theta_2-\theta_1$ between the two production costs and (ii) the exogenous distribution $\pi$ over production costs. In particular,  
\begin{enumerate}
    \item When the difference between $\theta_1$ and $\theta_2$ is small in the sense that $\theta_1$ and $\theta_2$ violate (\ref{costdifference}), the buyer offers a high price $p_{\theta_2}$ and the seller accepts immediately. 
The same limiting equilibrium arises when $\theta_1$ and $\theta_2$ satisfy (\ref{costdifference}) and the low type occurs with probability less than $\pi^*$.
    \item When $\theta_1$ and $\theta_2$ satisfy (\ref{costdifference}) and $\pi(\theta_1)> \pi^*$, the buyer offers a low price $\min \{p_{\theta_1},\theta_2\}$. The high type demands the entire surplus $1$ and the buyer concedes after some delay. This leads to an expected welfare loss of (\ref{expecteddelay}). The low type trades immediately either by accepting the buyer's offer or by offering $1-(\theta_2-\theta_1)$, depending on the comparison between $p_{\theta_1}$ and $\theta_2$.
\end{enumerate}

Theorem \ref{Theorem1} suggests that costly delays arise in equilibrium \textit{if and only if} the difference between the two production costs is large enough and the seller is likely to have a low production cost. Our inefficient bargaining result stands in contrast to the efficiency results in reputational bargaining games \textit{without} private payoff information (Kambe 1999 and Abreu and Gul 2000), as well as those in reputational bargaining games with one-sided private information about payoffs, but either the private information is about the discount rate (e.g., APS) or the set of commitment types violates our richness assumption (e.g., Fanning 2023). We discuss those models by the end of this section.

We argue that inefficient delays occur whenever the \textit{uninformed player}, i.e., the buyer, uses her offer to \textit{screen} the informed seller, that is, to induce sellers with different costs to demand different prices. Screening is \textit{feasible} when players can build reputations and the uninformed buyer faces uncertainty about her opponent's cost. Screening is \textit{profitable} for the uninformed buyer when both 
the probability of the low-cost type and 
the difference between the two costs are large enough.

In order to understand the intuition behind Theorem \ref{Theorem1}, we start from an auxiliary game where \textit{both types of the seller are required to demand the same price}. 
Suppose players' offers $p_b$ and $p_s$ satisfy $\theta_2<p_b<p_s<1$. In order to make the buyer indifferent between conceding and not conceding, the seller needs to concede at rate $\lambda_s \equiv \frac{r(1-p_s)}{p_s-p_b}$.
Since the seller's benefit from conceding is \textit{decreasing} in his production cost, the high-cost type will start to concede only after the low-cost type has finished conceding. Let $T_i$ denote the time at which type $\theta_i$ finishes conceding, and let $\hat{\varepsilon}$ denote the probability that the seller is committed conditional on offering $p_s$. By definition, 
\begin{equation}\label{time}
e^{-\lambda_s T_1} = \hat{\varepsilon} + \pi(\theta_2) (1-\hat{\varepsilon}) \textrm{ and }
e^{-\lambda_s (T_2-T_1)} = \frac{\hat{\varepsilon}}{\hat{\varepsilon} + \pi(\theta_2) (1-\hat{\varepsilon})}
\end{equation}
The buyer first concedes at rate $\lambda_b^1 \equiv \frac{r(p_b-\theta_1)}{p_s-p_b}$, which makes the low-cost type indifferent between conceding and not conceding. After the low-cost type finishes conceding at time $T_1$, the buyer concedes at a lower rate $\lambda_b^2 \equiv \frac{r (p_b-\theta_2)}{p_s-p_b}$, which makes the high-cost type indifferent. 
As $\varepsilon \rightarrow 0$, equation (\ref{time}) implies that $T_1$ is bounded above while $T_2$ diverges to $+\infty$. As a result, 
the buyer spends most of her time conceding at rate $\lambda_b^2$, so her time-average concession rate is close to $\lambda_b^2$.

Similar to Abreu and Gul (2000), both players face the trade-off that demanding more surplus will lower their concession rate and as $\varepsilon \rightarrow 0$, having a lower concession rate implies that they need to concede at time $0$ with probability close to $1$. Hence, the buyer can secure herself a payoff of approximately $1-p_{\theta_2}$ by offering $p_{\theta_2}$ and the seller with cost $\theta$ can secure himself a payoff of approximately $p_{\theta_2}-\theta$ by demanding $p_{\theta_2}$. The sum of players' secured payoffs equals the total surplus, which implies that their equilibrium payoffs are close to their secured payoffs. When the conditions in the first statement of Theorem \ref{Theorem1} are satisfied, this equilibrium in the auxiliary game remains an equilibrium in the original game and is the unique limit point 
as $\varepsilon$ and $\nu$ go to $0$.

When different types of the seller can offer \textit{different} prices, the buyer can \textit{screen} the seller by making a \textit{screening offer} that belongs to $(\theta_1,\theta_2]$. Compared to offering $p_{\theta_2}$ and inducing both types of the seller to trade immediately, each screening offer can induce different types of the seller to demand different prices. Although screening causes delays, the buyer may end up paying a lower price to the low-type seller, which is why she has an incentive to make such screening offers.

To elaborate, suppose the buyer offers $p_b \in (\theta_1,\theta_2]$. 
The low-cost seller obtains a strictly positive payoff from conceding, so he faces the usual trade-off identified in Abreu and Gul (2000) that demanding a higher price will lower his speed of reputation building. However, this trade-off is no longer relevant for the high-cost type since his payoff from conceding is non-positive. As a result, the high-cost type never benefits from conceding to the buyer and he prefers to demand a larger share of the surplus as long as doing so can convince the buyer that he has a high  cost.

The key step of our proof is stated as Lemma \ref{Lemma_4.5} in Appendix \ref{sec4}, that under our richness assumption on the set of commitment types, the high type must demand $1$ after the buyer makes any screening offer. This lemma also applies when the cost distribution is endogenous.

For some intuition, suppose by way of contradiction that the high type demands $p_s<1$ with strictly positive probability. The rational low-cost type must demand every price in  $\mathbf{P}_s \bigcap (p_s,1)$ with strictly positive probability. This is because otherwise, the buyer will assign probability $1$ to the high-cost type and the commitment type after observing the seller offering such a price, after which she will concede immediately. Anticipating this, the high-cost type will 
find it strictly profitable to deviate to one of these prices in $\mathbf{P}_s \bigcap (p_s,1)$ instead of offering $p_s$. We 
then establish a monotonicity result,  that for each pair of prices offered by the seller with positive probability, offering the higher price will lead to a longer expected delay and a higher expected trading price. This implies that when the low type has an incentive to demand a higher price $p_s' > p_s$, the high type should have no incentive to demand a strictly lower price $p_s$. Since the low type must demand every price in $\mathbf{P}_s \bigcap (p_s,1)$ with strictly positive probability and our richness assumption requires that $\sup \mathbf{P}_s \backslash \{1\}=1$, the high type must demand $1$ following any of the buyer's screening offer.

In summary, the buyer faces a trade-off when she chooses between making a screening offer $p_b \in (\theta_1,\theta_2]$ and the pooling offer $p_{\theta_2}$: Screening wipes out the surplus she can extract from the high type but may lower the price she pays to the low type. The latter is true if and only if the difference between $\theta_1$ and $\theta_2$ is large enough. This is because when $\theta_1$ and $\theta_2$ are too close, every screening offer $p_b \in (\theta_1,\theta_2]$ is too low relative to the Rubinstein bargaining price of the low type $p_{\theta_1}$. If $\theta_2+p_{\theta_2} \leq 2 p_{\theta_1}$ or equivalently $\theta_2-\theta_1 \leq \frac{1-\theta_2}{2}$, then for any $p_b \in (\theta_1,\theta_2]$, 
the low type can offer something greater than $p_{\theta_2}$ and induce the buyer to concede almost immediately, in which case screening is unprofitable for the buyer. This explains the logic behind (\ref{costdifference}). When $\theta_2-\theta_1 > \frac{1-\theta_2}{2}$, $\pi^*$ is the probability of the low type under which the buyer's benefit from screening equals her cost of screening, so the buyer prefers to make the screening offer if and only if $\pi(\theta_1) > \pi^*$.\footnote{When $\pi(\theta_1)=\pi^*$, although the buyer will be indifferent between the pooling offer $p_{\theta_2}$ and her optimal screening offer $p_b \in (\theta_1,\theta_2]$ \textit{in the limit} where $\varepsilon \rightarrow 0$, she will have a strict preference for one of these offers under generic $\varepsilon$. The cutoff $\pi^*$ will play a role in the game with endogenous technology adoption analyzed in the next section.}

In the last step, we compute
the expected delay $\pi(\theta_2) \Big\{ 1-\mathbb{E}[e^{-r \tau_b}|\theta=\theta_2] \Big\}$ and the expected welfare loss from delay $\pi(\theta_2) (1-\theta_2) \Big\{ 1-\mathbb{E}[e^{-r \tau_b}|\theta=\theta_2] \Big\}$ in the inefficient equilibria
by bounding the value of $\mathbb{E}[e^{-r \tau_b}|\theta=\theta_2]$.
Our bounds are derived via the two types of the seller's incentive constraints.
First, after the buyer makes a screening offer, type $\theta_1$ cannot find it profitable to demand $1$ in any equilibrium, 
 which leads to the following upper bound for $\mathbb{E}[e^{-r \tau_b}|\theta=\theta_2]$:
\begin{equation}\label{upperbound}
\underbrace{(1-\theta_1)\mathbb{E}[e^{-r \tau_b}|\theta=\theta_2] }_{\textrm{type $\theta_1$'s payoff from demanding $1$}} \leq \underbrace{(1+\theta_1-\min\{p_{\theta_1},\theta_2\})-\theta_1}_{\textrm{type $\theta_1$'s equilibrium payoff}}.
\end{equation}

Second, type $\theta_2$ cannot profit from demanding any $p_s$ that is strictly less than but close to $1$ and then never 
conceding to his opponent. In order to formally state this incentive constraint, we start by introducing a few extra notation. 
Fix players' offers $p_b=\min\{p_{\theta_1},\theta_2\}$ and $p_s$. Let $T_1$ denote the time it takes for type $\theta_1$ to finish conceding and let $c_b$ denote the probability with which the buyer concedes at time $0$, both of which depend on the buyer's posterior belief about the seller's type.
Type $\theta_2$'s incentive constraint implies that
\begin{align}\label{lowerbound}
\underbrace{(1-\theta_2)\mathbb{E}[e^{-r \tau_b}|\theta=\theta_2]}_{\textrm{type $\theta_2$'s equilibrium payoff}} \geq \underbrace{(p_s-\theta_2)  \Bigg(c_b+(1-c_b)\Big(1-e^{-(r+\lambda_b^{1})T_1} \Big)\frac{\min\{p_{\theta_1},\theta_2\}-\theta_1}{p_s-\theta_1}\Bigg)}_{\textrm{type $\theta_2$'s payoff from deviating to $p_s \approx 1$}}.
\end{align}
We show in Appendix \ref{sub4.1} that, as $p_s \rightarrow 1$ and $\varepsilon \rightarrow 0$, the right-hand-side of (\ref{lowerbound}) is at least
\begin{equation}\label{limitinglowerbound}
\frac{1-\min\{p_{\theta_1},\theta_2\}}{1-\theta_1} (1-\theta_2).
\end{equation}
Therefore, the upper and the lower bounds for $\mathbb{E}[e^{-r \tau_b}|\theta=\theta_2]$
coincide in the limit, which pin down the limiting value of
$\mathbb{E}[e^{-r \tau_b}|\theta=\theta_2]$. 
Our calculation also suggests that compared to the efficient equilibrium, the high-cost seller's payoff is weakly greater in the inefficient equilibrium. Therefore, the low-cost seller not only bears the ex-ante welfare loss from the buyer's screening but is also expropriated by the buyer.

\paragraph{Comparative Statics:} We apply Theorem \ref{Theorem1} to examine how the expected welfare loss and the expected delay of reaching agreement depend on the primitives, such as the distribution of the seller's production cost $\pi$, his production cost under the new technology $\theta_1$, and that under the default technology $\theta_2$. We measure the expected delay by $1-\mathbb{E}[e^{-r \min \{\tau_s,\tau_b\}}]$.
As in Theorem \ref{Theorem1}, we focus on the limit where $(\varepsilon,\nu) \rightarrow (0,0)$. We start by examining the effect of an increase in the fraction of sellers with a low production cost. 
\begin{Corollary}\label{cor1}
For every $\eta>0$, there exists $\bar{\nu}>0$ such that when $\nu<\bar{\nu}$, there exists $\bar{\varepsilon}_{\nu} >0$ such that for every $\varepsilon \in (0,\bar{\varepsilon}_{\nu})$, both the expected welfare loss and the expected delay are less than $\eta$ when $\pi (\theta_1) \in [0, \pi^*)$, but are strictly positive and are strictly decreasing in $\pi(\theta_1)$ when $\pi(\theta_1) \in (\pi^*,1)$. 
\end{Corollary}
Corollary \ref{cor1} suggests that the expected welfare loss from delay is maximized when $\pi(\theta_1)$ is slightly above $\pi^*$. 
Intuitively, bargaining is efficient when the low type occurs with probability no more than $\pi^*$. 
When $\pi(\theta_1)$ is above $\pi^*$, inefficient delay occurs only when the seller has a high production cost $\theta_2$, and 
conditional on $\theta=\theta_2$, the expected welfare loss from delay is independent of $\pi(\theta_1)$. 
Next, we examine the effect of an increase in the production cost of the low-cost type. 
\begin{Corollary}\label{cor2}
There exists $\bar{\nu}>0$ such that when $\nu<\bar{\nu}$, there exists $\bar{\varepsilon}_{\nu} >0$ such that for every $\varepsilon \in (0,\bar{\varepsilon}_{\nu})$, both the expected delay and the expected welfare loss are weakly decreasing in $\theta_1$.
\end{Corollary}
Intuitively, improving the efficiency of the new technology (i.e., a decrease in $\theta_1$) has two effects, both of which 
lead to a longer expected delay. 
First, a lower $\theta_1$ makes screening more profitable for the buyer, which expands the range of $\pi$ under which the buyer prefers to make the screening offer. Second,  when $\pi(\theta_1)>\pi^*$,
the expected delay after the high type offers $1$ weakly increases as $\theta_1$ decreases, and strictly increases whenever $\theta_2 \leq p_{\theta_1}$. This is driven by the two incentive constraints that pin down the expected delay: the low type's incentive constraint leads to a lower bound on the expected delay and the high type's incentive constraint leads to an upper bound.  
According to (\ref{upperbound}) and (\ref{lowerbound}), as $\theta_1$ decreases, the low type's gain from deviation increases and 
the high type's gain from deviation decreases. 
Hence, 
the expected delay that satisfies both incentive constraints increases. 
Next, we examine the effect of an increase in the production cost of the high-cost type. 
\begin{Corollary}\label{cor3}
There exists $\bar{\nu}>0$ such that when $\nu<\bar{\nu}$, there exists $\bar{\varepsilon}_{\nu} >0$ such that for every $\varepsilon \in (0,\bar{\varepsilon}_{\nu})$, the expected delay is weakly increasing in $\theta_2$, and the expected welfare loss is weakly increasing in $\theta_2$ when $\theta_2\in(\theta_1,p_{\theta_1})$ and is weakly decreasing in $\theta_2$ when $\theta_2 \in (p_{\theta_1},1)$.
\end{Corollary}
Intuitively, improving the efficiency of the default technology (i.e., a decrease in $\theta_2$) has two effects.
First, a lower $\theta_2$ makes screening less profitable, which reduces the range of $\pi$ under which the buyer prefers to make the screening offer. This decreases the expected delay as well as the expected welfare loss from delay. However, there is another effect, namely, a lower $\theta_2$ increases the surplus from trading with type $\theta_2$, which makes each unit of delay more costly in terms of social welfare. Overall, players will reach an agreement sooner when the default technology becomes more efficient, and the expected welfare loss also decreases if and only if $\theta_2$ is lower than $p_{\theta_1}$.

\paragraph{Remarks:} 
Theorem \ref{Theorem1} shows that bargaining is inefficient when (i) the set of commitment types is \textit{rich} and (ii) the cost difference between the two types of the seller is large enough. 
Section \ref{sec5}
extends these findings to cases with three or more production costs. The presence of bargaining inefficiencies stands in contrast to APS and Fanning (2024). APS assume that players only have private information about their discount rate, in which case there is no offer under which some type has a strict incentive to concede while other types have no incentive to concede. This explains why the uninformed player cannot induce different types of the informed player to offer different prices. 

Fanning (2024) studies the interaction between private information about values and outside options in reputational bargaining models with exogenous cost/value distributions. This is related to our analysis in the current section. 
He assumes that no rational type is indifferent between accepting and rejecting any commitment type's offer. This rules out for example, a commitment type who demands the entire surplus, which is required to \textit{exist} for our result. The motivation for our requirement is explained in Section \ref{sub2.1}. 
Intuitively, when it is \textit{infeasible} for the seller to build a reputation for demanding the entire surplus, the buyer \textit{cannot} screen the seller in any equilibrium since it is not optimal for her to concede after delay knowing that the seller will never concede.

\subsection{Reputational Bargaining with Endogenous Technology Adoption}\label{sub3.2}
This section analyzes the reputational bargaining game in which the seller's production cost is \textit{endogenously} determined by his adoption decision before the bargaining stage and the buyer \textit{cannot} observe whether he has adopted. Recall the definition of $\pi^*$ in (\ref{cutoff}) and that $\pi^*<1$ if and only if $(\theta_1,\theta_2)$ satisfies (\ref{costdifference}). If $(\theta_1,\theta_2)$ also satisfies a stronger condition that $p_{\theta_1}< \theta_2$, then $\pi^*= \frac{1-\theta_2}{1-\theta_1}$. 

\begin{figure}[h]
    \centering
    \includegraphics[width=.7\textwidth]{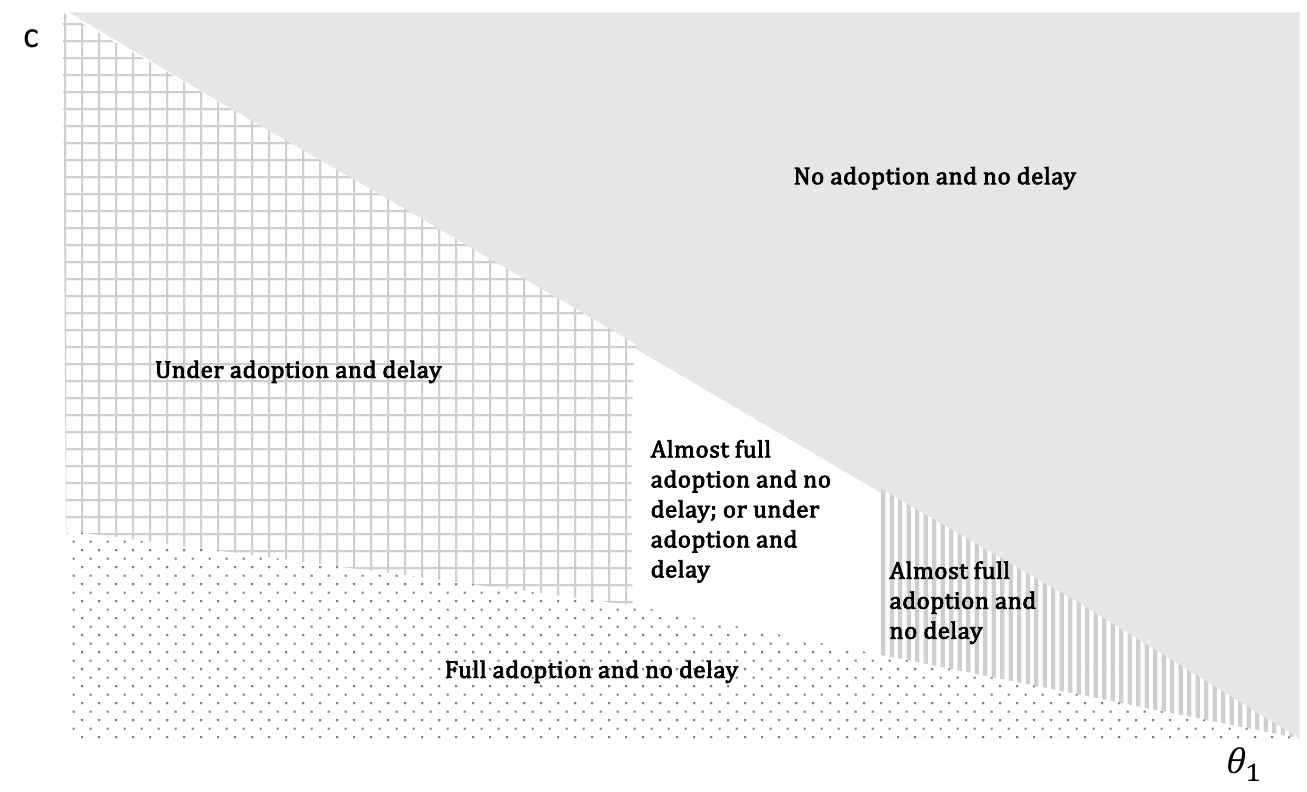}
\caption{Limiting equilibria in the reputational bargaining game with endogenous technology adoption when $\theta_2>1/3$. Aside from the white region, there is a unique limiting equilibrium. In the white region, there are two possible limiting equilibria; one has efficient investment and negligible delay in reaching agreement, and another has inefficient investment and significant delay in reaching agreement. \textit{Almost full adoption and no delay} means that the limiting outcome is efficient but for every $\varepsilon>0$, the probability of adoption is strictly less than one and the expected delay is strictly positive. If $\theta_2<1/3$, the region in which an inefficient equilibrium exists is empty.}
 \label{fig:eqm}
\end{figure}
We state the interesting parts of our equilibrium characterization as Theorem \ref{Theorem2}. A more detailed description can be found 
in Lemmas \ref{lemmab2}-\ref{lemmab5} in Appendix \ref{sub4.2}, which we depict graphically in Figure  \ref{fig:eqm}.
\begin{Theorem}\label{Theorem2}
There exists at least one equilibrium. For every $\eta>0$, there exists $\bar{\nu}>0$ such that when $\nu<\bar{\nu}$, there exists $\bar{\varepsilon}_{\nu} >0$ such that for every $\varepsilon \in (0,\bar{\varepsilon}_{\nu})$:
\begin{enumerate}
    \item Suppose $(\theta_1,\theta_2)$ violates (\ref{costdifference}). In every equilibrium, the expected delay is less than $\eta$, and the adoption probability is less than $\eta$ if  $c> \theta_2-\theta_1$ and is more than $1-\eta$ if $c< \theta_2-\theta_1$. 
    \item Suppose $(\theta_1,\theta_2)$ satisfies (\ref{costdifference}). There exists a non-empty open interval $(\underline{c},\overline{c}) \subset \Big( \frac{\theta_2-\theta_1}{2},\theta_2-\theta_1 \Big)$ such that for every $c \in (\underline{c},\overline{c})$, there exists an equilibrium where the adoption probability is within an $\eta$-neighborhood of $\pi^*$, and the expected delay is bounded above $0$.
    \item Suppose $(\theta_1,\theta_2)$ satisfies $p_{\theta_1}< \theta_2$. If $c \in \Big(\frac{\theta_2-\theta_1}{2},\theta_2-\theta_1 \Big)$, then in all equilibria, the adoption probability is within an $\eta$-neighborhood of $\pi^*$, and the expected delay is bounded above $0$.
\end{enumerate}
\end{Theorem}
The proof is in Appendix \ref{sub4.2}. 
Theorem \ref{Theorem2} shows that when the probability of commitment types vanishes, the seller's adoption decision can be socially inefficient \textit{if and only if} the social benefit from adoption $\theta_2-\theta_1$ is large enough. Under a stronger condition that $\theta_2$ is greater than the Rubinstein bargaining price under the production cost of the new technology $p_{\theta_1} \equiv \frac{1+\theta_1}{2}$, inefficient adoption and costly delay occur in \textit{all} equilibria as long as the adoption cost $c$ is between half of the social benefit from adoption $\frac{\theta_2-\theta_1}{2}$ and the entire social benefit from adoption $\theta_2-\theta_1$.

Theorem \ref{Theorem2} has two implications. First, inefficient adoption occurs in equilibrium if and only if  \textit{in expectation}, there are significant delays in reaching agreement. This stands in contrast to the benchmark scenario where the buyer can observe the seller's adoption decision, in which case there is negligible delay in reaching agreement but the seller's adoption decision is socially inefficient. Second, an increase in the social benefit from adoption $\theta_2-\theta_1$ can lower the probability of adoption. This is because inefficient adoption and costly delays can arise only when the social benefit from adoption is large enough. We will formally state this finding as Corollary \ref{cor4} later in this section.

One useful observation that is \textit{not} explicitly stated in Theorem \ref{Theorem2} is that when the parameter values belong to the striped region of Figure \ref{fig:eqm}, the buyer will offer something close to the Rubinstein bargaining price of the low-cost type $p_{\theta_1}$ and costly delay will arise in \textit{all} equilibria \textit{conditional on} the seller having a high production cost. This is the case even when $\theta_1$ and $\theta_2$ violate (\ref{costdifference}), under which for \textit{any exogenous full support distribution} over production costs, the buyer will offer something close to $p_{\theta_2}$
and there is almost no delay in reaching an agreement.

In order to understand Theorem \ref{Theorem2} as well as how to square the above observation with Theorem \ref{Theorem1}, we start from a heuristic explanation for Theorem \ref{Theorem2} using our result for reputational bargaining games with an \textit{exogenous} cost distribution. Then we point out a contradiction that results from this line of reasoning and explain why Theorem \ref{Theorem1} \textit{cannot} be directly applied to settings where the cost distribution is \textit{endogenous}. Then we provide the correct explanation and explain the differences between games with endogenous adoption and ones with exogenous cost distributions.

First, let us examine the seller's equilibrium share of the social surplus from technology adoption. 
Fix any $\pi \in \Delta \{\theta_1,\theta_2\}$ that has full support and satisfies $\pi(\theta_1) \neq \pi^*$. Theorem \ref{Theorem1} implies that as $\varepsilon \rightarrow 0$ (i) in every efficient equilibrium,
the difference between the low-cost type's equilibrium payoff and that of the high-cost type's is approximately $\theta_2-\theta_1$, and (ii) in every inefficient equilibrium, the difference in equilibrium payoffs is approximately 
$ (\theta_2-\theta_1) \alpha$ where
\begin{displaymath}
\alpha \equiv \left\{ \begin{array}{ll}
\frac{1}{2} & \textrm{if } p_{\theta_1}< \theta_2\\
\frac{1-\theta_2}{1-\theta_1} & \textrm{if } p_{\theta_1} \geq \theta_2 \textrm{ and } (\theta_1,\theta_2) \textrm{ satisfies } (3.5).
\end{array} \right.
\end{displaymath}
Intuitively, in every efficient equilibrium, both types of the seller trade immediately at the same price, in which case the seller captures the entire surplus from adoption. In every inefficient equilibrium, 
the seller only captures part of the surplus due to the buyer's screening offer.

Next, let us examine the seller's equilibrium adoption probability when his 
adoption cost $c$ is \textit{strictly} between  $(\theta_2-\theta_1) \alpha$ and $\theta_2-\theta_1$. The seller cannot adopt with zero probability since the buyer will offer a high price $p_{\theta_2}$. If this is the case, then the seller's gain from adoption is $\theta_2-\theta_1$, which will provide him a \textit{strict} incentive to adopt the technology. 
This contradicts the hypothesis that he adopts with zero probability. 
He cannot adopt for sure since the buyer will then offer $p_{\theta_1}$ in which case type $\theta_2$'s payoff is at least $\frac{1-\theta_2}{2}$ when he demands $p_s \approx 1$ and never concedes. The seller's gain from adoption is at most $\frac{\theta_2-\theta_1}{2}$. As long as the cost of adoption is strictly more than that, the seller will have no incentive to adopt  which leads to a contradiction. 
We can also rule out interior adoption probabilities that are not $\pi^*$, 
since the seller's gain from adoption will not equal his adoption cost, in which case he will have no incentive to mix at the adoption stage.

Therefore, we have ruled out all adoption probabilities except for $\pi^*$. 
However, when $(\theta_1,\theta_2)$ violates (\ref{costdifference}), or equivalently $\pi^*=1$, if
the seller adopts with probability $\pi^*=1$ and his adoption cost satisfies $(\theta_2-\theta_1)\alpha<c< \theta_2-\theta_1$, then our earlier reasoning suggests that his benefit from adoption will be strictly lower than his cost, which again leads to a contradiction. This contradiction seems to rule out all possible adoption probabilities under certain parameter values.

Such a contradiction arises due to a mistake in the above argument, that is, Theorems \ref{Theorem1} and \ref{Theorem2} are different in terms of their orders of taking limits, making Theorem \ref{Theorem1} and other existing results on reputational bargaining inapplicable to settings where $\pi$ is endogenous. Specifically, Theorem \ref{Theorem1} characterizes the set of equilibria under a \textit{fixed cost distribution $\pi$} in the limit where $\varepsilon \rightarrow 0$. The same order of limits also applies to the results in APS and Fanning (2024). However, 
$\pi$ depends on the probability of commitment types $\varepsilon$ when the adoption probability is \textit{endogenous} and the probability that the seller has a high production cost may also vanish as $\varepsilon \rightarrow 0$.

%\footnote{A related issue appears in Gul (2001) who studies Coasian bargaining with endogenous investments. Unlike Gul, Sonnenschein and Wilson (1986) who fix the distribution of values and then send the frequency of offers to infinity, the informed player's investment probability depends on the frequency of offers in Gul (2001).} 

This is indeed what happens when $(\theta_1,\theta_2)$ violates (\ref{costdifference}) and  $c \in \big(\frac{\theta_2-\theta_1}{2}, \theta_2-\theta_1 \big)$. We show that for every $\xi>0$, there exists $\overline{\varepsilon}>0$ such that Theorem \ref{Theorem1} applies for all $\varepsilon< \overline{\varepsilon}$ and $\pi$ with $\pi(\theta_1) \in [0,1-\xi]$. However, the qualitative features of the equilibria are different for any fixed $\varepsilon >0$ as $\pi(\theta_1) \rightarrow 1$. 
Although for any fixed $\pi \in \Delta \{\theta_1,\theta_2\}$, the buyer strictly prefers the pooling offer $p_{\theta_2}$ as $\varepsilon \rightarrow 0$, she will prefer to offer $p_{\theta_1}$ for any small but fixed $\varepsilon>0$ as $\pi(\theta_1)$ goes to $1$. In response to the buyer's pooling offer, type-$\theta_2$ seller will demand a higher price and trade with delay, with the expected delay pinned down by the seller's indifference condition at the adoption stage. Type-$\theta_1$ seller will accept the buyer's offer with probability close to $1$, and pool with type $\theta_2$ with probability close to $0$. Although there are significant delays conditional on the seller's production cost being $\theta_2$, these delays have negligible payoff consequences from an ex ante perspective since the probability with which the seller's production cost is $\theta_2$ will go to $0$ as $\varepsilon \rightarrow 0$.

On the other hand, when $c \in \big(\frac{\theta_2-\theta_1}{2}, \theta_2-\theta_1 \big)$ and $(\theta_1,\theta_2)$ satisfies not only (\ref{costdifference}) but also a stronger condition that $p_{\theta_1}<\theta_2$, one can no longer sustain an approximately efficient outcome where the seller adopts the technology with probability close to $1$. This is because when the buyer offers $p_{\theta_1}$, the seller has no incentive to concede when his cost is $\theta_2$. As a result, the seller can secure a payoff of $\frac{1-\theta_2}{2}$ by not adopting the technology and demanding something close to $1$. This guaranteed payoff $\frac{1-\theta_2}{2}$ is strictly greater than his payoff from adopting the technology and accepting the buyer's offer $p_{\theta_1}$, which contradicts the hypothesis that the seller adopts the technology with probability close to $1$. Therefore, in every equilibrium, the seller will adopt with probability bounded below $1$.  As $\varepsilon \rightarrow 0$, the equilibrium adoption probability is close to $\pi^*$, since it is the only adoption probability that can make the buyer indifferent between the pooling offer $p_{\theta_2}$ and her optimal screening offer.

When $(\theta_1,\theta_2)$ satisfies (\ref{costdifference}) but $p_{\theta_1} \geq \theta_2$, and $c \in (\underline{c},\overline{c})\equiv\big(\frac{(1-\theta_2)(\theta_2-\theta_1)}{1-\theta_1}, \theta_2-\theta_1 \big)$, there exist inefficient equilibria where the seller adopts with probability close to $\pi^*$ since Theorem \ref{Theorem1} applies uniformly to all $\pi$ with $\pi(\theta_1)$ bounded below $1$. However, there may also exist efficient equilibria where the seller adopts with probability close to $1$. We explain why there can be multiple limit points in Appendix \ref{sub4.2}.

The above explanation also sheds light on why inefficient adoption occurs \textit{if and only if} there are significant delays in bargaining. Intuitively, if the equilibrium adoption decision is inefficient, then it cannot be the case that both types of the seller trade immediately. This is because otherwise, both types must trade at the same price and the seller can capture all the gains from adoption, providing him a strict incentive to make the efficient adoption decision. If the adoption decision is approximately efficient, then the probability with which the seller does not adopt must be arbitrarily close to zero when $c<\theta_2-\theta_1$.\footnote{In our model, inefficient adoption is caused by the hold-up problem, so the seller will not adopt when his adoption cost is greater than the social benefit. This implies that inefficiency can only take the form of under-adoption.} 
Since inefficient delay \textit{cannot} occur when the seller has a low cost, the welfare losses from delay must be negligible from an ex ante perspective.

%The other direction, namely that inefficient technology adoption implies inefficient bargaining, is also closely related to the buyer's incentive to screen. We show that, in any equilibrium with inefficient adoption, the adoption probability is interior. For the seller to be willing to mix between the two technologies, the outcome at the bargaining stage must feature some delay: If not, an agreement is reached immediately at a fixed price, and the seller would strictly prefer one technology over the other. By the arguments in Theorem \ref{Theorem1}, delay arises when the buyer screens, so it follows that screening must be a part of any equilibrium where the adoption decision is inefficient. This also highlights why (\ref{costdifference}) is a necessary condition for the existence of an inefficient equilibrium: Suboptimal adoption can only be sustained when the buyer has the incentive to screen, which is always ruled out when \ref{costdifference} is violated. 

\paragraph{Adoption Probability \& Welfare:} We provide sufficient conditions under which an increase in the social benefit from adoption lowers the probability of adoption and leads to a longer expected delay. 
\begin{Corollary}\label{cor4}
For every $\theta_1$, $\theta_2$, and $c$ that satisfy
\begin{equation*}
\theta_2-\theta_1>\frac{1-\theta_2}{2}, \text{ and } \max\Big\{\frac{1}{2},\frac{1-\theta_2}{1-\theta_1}\Big\}(\theta_2-\theta_1) <c< \theta_2-\theta_1,
\end{equation*}
and every $\widehat{\theta}_1< \theta_1$ that satisfies $\frac{1+\widehat{\theta}_1}{2}<\theta_2$ and $\widehat{\theta_1} \in (\theta_2-2c,\theta_2-c)$,
there exists $\overline{\nu}>0$ such that for every $\nu < \overline{\nu}$, there exists $\overline{\varepsilon}_\nu>0$ such that if $\varepsilon<\overline{\varepsilon}_\nu$,
\begin{enumerate}
    \item The probability of adoption in any equilibrium under $(\theta_1,\theta_2,c,\varepsilon,\nu)$ is strictly greater than the probability of adoption in any equilibrium under $(\widehat{\theta}_1,\theta_2,c,\varepsilon,\nu)$. 
\item The expected delay in any equilibrium under $(\theta_1,\theta_2,c,\varepsilon,\nu)$ is strictly less than the expected delay in any equilibrium under $(\widehat{\theta}_1,\theta_2,c,\varepsilon,\nu)$. 
\item The expected welfare in any equilibrium under $(\theta_1,\theta_2,c,\varepsilon,\nu)$ is weakly greater than the expected welfare in any equilibrium under $(\widehat{\theta}_1,\theta_2,c,\varepsilon,\nu)$. 
\end{enumerate}
\end{Corollary}
\begin{figure}[h]
    \centering
    \subfloat[Adoption probability]{\includegraphics[width=.45\textwidth]{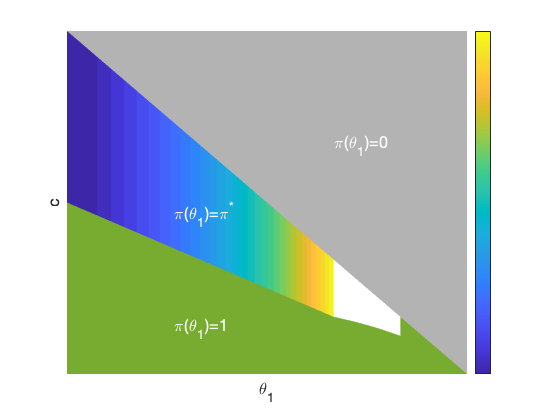} \label{fig:adoption}}
    \subfloat[Expected delay]{\includegraphics[width=.45\textwidth]{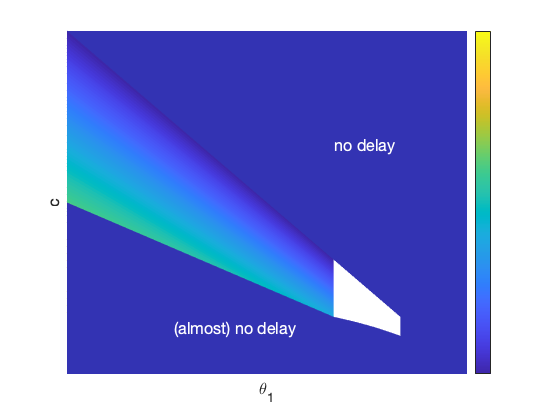} \label{fig:delay}}
    \caption{Comparative statics on the equilibrium outcomes. The white region represents the set of parameter values under which there are multiple limiting equilibria. In regions where the unique limiting equilibrium is inefficient, the values of  $\pi(\theta_1)$ and $1-\mathbb{E}[e^{-r \min \{\tau_s,\tau_b\}}]$ are depicted in panels (a) and (b), respectively, in ascending vertical order according to the color bar on the right of the figure. 
   That is, aside from the white region, a deeper color represents a lower adoption probability (in the left panel) and a lower expected delay (in the right panel) in the unique limiting equilibrium. 
    In the remaining regions, the seller's adoption decision is socially efficient: The adoption probability is $0$ when $c>\theta_2-\theta_1$, and is $1$ when $c<\theta_2-\theta_1$, and the expected delay is zero in the limit.}
\end{figure}
The proof is in Online Appendix F. We depict the complete comparative statics on the adoption probability and the expected delay 
in Figures \ref{fig:adoption} and \ref{fig:delay}, where the white region represents parameter values under which there are multiple limiting equilibria. Corollary \ref{cor4} implies that when the production cost under the new technology decreases from $\theta_1$ to $\widehat{\theta}_1$, i.e., adoption becomes more socially beneficial, the probability of adoption decreases as long as $\theta_2-\widehat{\theta}_1$ is intermediate: It is large enough so that the buyer has an incentive to screen the seller, but is not too large relative to the adoption cost $c$ so that the seller has no incentive to adopt if he knew that the buyer will offer $p_{\widehat{\theta}_1}$. It also implies that a decrease from $\theta_1$ to $\widehat{\theta}_1$ when $\theta_2-\widehat{\theta}_1$ is intermediate can also lead to a longer expected delay, which leads to further efficiency losses.

%If the cost of adoption decreases from $c$ to $\widehat{c}$, except for the parameter values under which there are multiple limit points, the probability of adoption weakly increases, which follows from Theorem \ref{Theorem2}. The effects on the expected delay is ambiguous. This is because there is no delay when $c> \theta_2-\theta_1$ or when $c< \frac{\theta_2-\theta_1}{2}$, in which cases the seller either never adopts the new technology or adopts the new technology for sure, so there is no delay in trade. In contrast, when $c \in (\frac{\theta_2-\theta_1}{2},\theta_2-\theta_1)$, there are inefficient equilibria in which the seller adopts with probability strictly between $0$ and $1$ and the buyer offers $p_{\theta_1}$ with positive probability, leading to significant delays in equilibrium. 

\section{Extension: Choosing Between Multiple New Technologies}\label{sec5}
This section extends our theorems to settings where the seller chooses a production technology from $\{1,2,...,n\}$ before bargaining with the buyer, where $\theta_j$ stands for the production cost of technology $j$ and $c_j$ stands for the cost of adopting technology $j$. Let $\Theta\equiv\{\theta_1,...,\theta_n\}$ and $C\equiv\{c_1,...,c_n\}$.

We assume that $0<\theta_1<...<\theta_n<1$ and $c_1>...>c_n=0$. This implies that (i) there exists a default technology $\theta_n$ that is costless to adopt, (ii) all new technologies $\theta_1,...,\theta_{n-1}$ are costly to adopt but lead to lower production costs compared to the default one, and (iii) technologies that have higher adoption costs have lower production costs. 
These assumptions are without loss of generality since a technology will never be adopted in any equilibrium if it costs more than a more efficient technology. We make a generic assumption that there is a unique \textit{socially efficient technology} and focus on the interesting case that the socially efficient technology is not the default one. Formally, we assume that there exists $j^o<n$ such that $\{j^o\} = \argmin_{k \in \{1,2,...,n\}} \Big\{ \theta_k+c_k \Big\}$.

First, we consider a reputational bargaining game where the distribution over production cost $\pi\in\Delta(\Theta)$ is exogenous. Theorem \ref{Theorem3} characterizes players' equilibrium strategies in the limit as $\nu$ and $\varepsilon$ go to zero. Let $\sigma^*_{b,i}\in\Delta[0,1]$ denote the buyer's strategy of offering $\min\{p_{\theta_i},\theta_{i+1}\}$. For $p_b>\theta_1$, let
\begin{align}\label{equilibriumoffers}
    p_{s,\theta}^*(p_b) \equiv 
    \begin{cases} 
    1, \quad & \text{if  } p_b\leq \theta, \\
    \max\big\{p_b,1+\max\{\hat{\theta}\in \Theta:p_b>\hat{\theta}\}-p_b\big\}, \quad & \text{if  } p_b>\theta, 
    \end{cases}
\end{align}
and let $\sigma_{s,\theta}^*(p_b)\in \Delta[0,1]$ be the strategy for type $\theta$ that assigns probability $1$ to $p_{s,\theta}^*(p_b)$. That is, for every $p_b>\theta_1$ and $\theta_j=\max\{\hat{\theta}\in\Theta:p_b>\hat{\theta}\}$, $\sigma^*_s\equiv (\sigma_{s,\theta}^*)_{\theta\in\Theta}$ prescribes all types with production cost strictly greater than $\theta_j$ to demand the entire surplus $1$, and all types with production cost no more than $\theta_j$ to offer a price under which the buyer and the seller have the same concession rate when the seller's production cost is known to be $\theta_j$.

For any $i,j\in \{1,...,n\}$ such that $i<j$, let $\pi[\theta_i,\theta_j]$ be the probability that $\theta\in [\theta_i,\theta_j]$. Theorem \ref{Theorem3} characterizes the unique limiting equilibrium of the reputational bargaining game with an exogenous cost distribution under the generic conditions that
\begin{equation} \label{optcutoff}
\argmax_{i\in\{1,...,n\}}\pi[\theta_1,\theta_i] \Big(\min\{p_{\theta_i},\theta_{i+1}\}-\theta_i \Big) 
\end{equation}
is a singleton with its unique element denoted by $i^*$, and that the cost distribution $\pi$ is interior. By an analogous logic to \ref{sub3.1}, $\pi[\theta_1,\theta_i] \Big(\min\{p_{\theta_i},\theta_{i+1}\}-\theta_i \Big)$ is the buyer's limiting payoff when she screens all sellers with type below $\theta_i$---i.e., when she follows the strategy $\sigma^*_{b,i}$. Thus, $i^*$ represents the uniquely optimal screening cutoff for the buyer.\footnote{We show this formally in the proof of Theorem \ref{Theorem3}.}
\begin{Theorem} \label{Theorem3}
There exists at least one equilibrium. Suppose $\pi \in \Delta (\Theta)$ is such that $\pi(\theta)>0$ for all $\theta\in\Theta$. For every $\eta>0$, there exists $\bar{\nu}>0$ such that when $\nu<\bar{\nu}$, there exists $\bar{\varepsilon}_{\nu} >0$ such that for every $\varepsilon \in (0,\bar{\varepsilon}_{\nu})$ and every equilibrium $(\sigma_s,\sigma_b,\tau_s,\tau_b)$ under $(\varepsilon, \nu)$,
\begin{enumerate}
    \item $\sigma_b$ is $\eta$-close to $\sigma_{b,i^*}$, and $\sigma_s(p_b)$ is $\eta$-close to $\sigma_{s}^*(p_b)$ for every $p_b$ such that $\sigma_b(p_b)>\eta$.
    \item Conditional on $\theta\leq \theta_{i^*}$, the expected welfare loss from delay is less than $\eta$. 
    \item Conditional on $\theta>\theta_{i^*}$, the buyer's payoff is $0$, and the expected welfare loss from delay is $\eta$-close to
\begin{equation}\label{expecteddelay2}
(1-\theta) \Bigg\{ 1-\frac{1-\min\Big\{p_{\theta_{i^*}},\theta_{i^*+1}\Big\}}{1-\theta_{i^*}}\Bigg\}. 
\end{equation}
\end{enumerate}
\end{Theorem}
The proof is in Online Appendix D, which is similar to the one for Theorem \ref{Theorem1} except that we need to establish a general version of Lemma \ref{Lemma_4.5} that allows for three or more production costs. The details can be found in Online Appendix C. 
According to Theorem \ref{Theorem3}, the qualitative features of the equilibrium in this general environment are similar to the ones in which there are two costs. By making an offer $p_b$ that belongs to $(\theta_i,\theta_{i+1}]$ with $i\in\{1,...,n-1\}$, the buyer is able to screen the seller by providing incentives to all types with cost weakly lower than $\theta_i$ to trade with negligible delay, and all types with cost strictly greater than $\theta_i$ to separate and demand the entire surplus. Using the same arguments as those in the proof of Theorem \ref{Theorem1}, the optimal way in which the buyer can screen types with cost no more than $ \theta_i$ is by offering approximately $\min\{p_{\theta_i},\theta_{i+1}\}$. This ensures the buyer a payoff of $\pi[\theta_1,\theta_i](\min\{p_{\theta_i},\theta_{i+1}\}-\theta_i)$ in the limit.

According to our theorem,
the buyer will screen the seller in equilibrium if and only if $i^* <n$, in which case she will offer a price that is lower than $\theta_n$ instead of offering a price that is close to $p_{\theta_n} \equiv \frac{1+\theta_n}{2}$. Conditional on the buyer offering $p_b$ in a $\nu$-neighborhood of $\min\{p_{\theta_{i^*}},\theta_{i^*+1}\}\leq \theta_n$, there will be inefficient delay whenever the seller's production cost satisfies $\theta>\theta_{i^*}$. This is because delay is necessary in order to satisfy the low-cost types' incentive constraints. The expected delay in (\ref{expecteddelay2}) is pinned down by the conditions that after the buyer offers $\min\{p_{\theta_{i^*}},\theta_{i^*+1}\}$, (i) type $\theta_{i^*}$ does not benefit from demanding $1$, and (ii) type $\theta_{i^*+1}$ does not profit from deviating to making an offer slightly below $1$ and waiting for the buyer to concede. 

In the complementary case in which $i^*=n$, the buyer strictly prefers to make a pooling offer  $p_{\theta_n}$ after which all types of the seller accept immediately, than to make any screening offer. 
As in the two-cost environment, trade happens immediately at a price of approximately $p_{\theta_n}$. 

Next, we describe the limiting equilibria in the game with \textit{endogenous} technology adoption. Under the assumption that $j^o<n$, the investment and the bargaining outcome may be inefficient provided that the buyer finds it beneficial to screen the seller. We provide a condition on the gap between different types of the seller's production costs under which the buyer may benefit from making screening offers \textit{under some distribution over cost types}. This is the case if and only if
\begin{align}
   \theta_n-\theta_1>\frac{1-\theta_{n}}{2}. \label{costdifference2}
\end{align}
As we explain later, if (\ref{costdifference2}) is violated, then the buyer can never benefit from screening the seller by making an offer below $\theta_{j^o}$, since any such offer is strictly dominated by offering $p_{\theta_{n}}$. 
\begin{Theorem} \label{Theorem4}
    There exists at least one equilibrium. For every $\eta>0$, there exists $\bar{\nu}>0$ such that when $\nu<\bar{\nu}$, there exists $\bar{\varepsilon}_{\nu} >0$ such that for every $\varepsilon \in (0,\bar{\varepsilon}_{\nu})$
\begin{enumerate} 
\item If (\ref{costdifference2}) is violated, the seller adopts $\theta_{j^o}$ with probability greater than $1-\eta$, and the expected welfare loss from delay is less than $\eta$, in any equilibrium.
\item If (\ref{costdifference2}) is satisfied, there exists an open set of adoption costs such that there exists an equilibrium where the seller adopts $\theta_{j^o}$ with probability bounded below one, and the expected delay is bounded above zero. 
\item If $p_{\theta_{j^o}}<\theta_n$, there exists an open set of adoption costs such that, in all equilibria, the seller adopts $\theta_{j^o}$ with probability bounded below one, and the expected delay is bounded above $0$. 
\end{enumerate}
\end{Theorem}
The proof is in Online Appendix E. For an intuitive description of Theorem \ref{Theorem4}, consider first the case in which (\ref{costdifference2}) is violated. Let 
$\pi\in\Delta(\Theta)$ be the distribution over production costs that the seller chooses in equilibrium. 
We relabel the elements of $\Theta$ so that $\text{supp}(\pi)=\{\hat{\theta}_1,...,\hat{\theta}_m\}$ with $\hat{\theta}_1<...<\hat{\theta}_m$. Let $\hat{c}_j$ denote the adoption cost of $\hat{\theta}_j$. If (\ref{costdifference2}) 
is violated, then 
\begin{equation*}
\hat{\theta}_{i+1}\leq \theta_1+1-p_{\theta_n}<\hat{\theta}_i+1-p_{\hat{\theta}_i}=p_{\hat{\theta}_i}
\textrm{ for every } i \in \{1,...,m-1\}.
\end{equation*}
As in the baseline model, the buyer may offer  $\hat{\theta}_{i+1}$ in order to screen the seller with cost less than $\hat{\theta}_i$.
If the buyer screens the seller by offering $\hat{\theta}_{i+1}$ for some $i\in \{1,...,m-1\}$, then her payoff is 
\begin{equation*}
\pi[\hat{\theta}_1,\hat{\theta}_{i}](\hat{\theta}_{i+1}-\hat{\theta}_i)\leq \pi[\hat{\theta}_1,\hat{\theta}_i](1-p_{\theta_n})\leq \pi[\hat{\theta}_1,\hat{\theta}_i](1-p_{\hat{\theta}_m})<1-p_{\hat{\theta}_m}.
\end{equation*}
This implies that any such offer is strictly dominated by offering $p_{\hat{\theta}_m}$. 
Given that the equilibrium price offered by the buyer is arbitrarily close to $p_{\hat{\theta}_m}$, the seller's equilibrium payoff converges to $p_{\hat{\theta}_m}-\hat{\theta}_m-\hat{c}_m\leq p_{\hat{\theta}_m}-\theta_{j^o}-c_{j^o}$, with strict inequality if $\hat{\theta}_m\neq \theta_{j^o}$. This implies that, if \eqref{costdifference2} is violated, the seller must adopt the socially efficient technology $j^o$ with probability converging to $1$. An argument analogous to the one in Theorem \ref{Theorem2} implies that the bargaining outcome conditional on the seller adopting the socially efficient technology must feature negligible delay.

Conversely, condition (\ref{costdifference2}) is sufficient for inefficiencies to arise in equilibrium under an open set of production costs. In particular, if (\ref{costdifference2}) is satisfied and the adoption costs are such that $j^o=1$, we can construct an equilibrium where the seller mixes between adopting $\theta_{j^o}$ and using the default technology $\theta_n$, in an analogous way as in Section \ref{sub3.2}. To do this, the seller's adoption strategy must be such that, in the limit, the buyer is indifferent between offering $p_{\theta_n}$ and $\min\{p_{\theta_{j^o}},\theta_n\}$, which yields that $\pi(\theta_{j^o})$ converges to $\frac{p_{\theta_n}-\theta_n}{\min\{p_{\theta_{j^o}},\theta_n\}-\theta_{j^o}}$. Condition (\ref{costdifference2}) implies that $\pi(\theta_{j^o})<1$. 

The discussions above imply that the seller with production cost $\theta_n$ trades with delay in equilibrium when the buyer offers $\min\{p_{\theta_{j^o}},\theta_n\}$. Moreover, if $c_{j^o}>\max\big\{\frac{1}{2},\frac{1-\theta_n}{1-\theta_{j^o}}\big\}(\theta_n-\theta_{j^o})$, then there exists a mixed strategy over bargaining postures for the buyer that assigns probability to $p_{\theta_n}$ and $\min\{p_{\theta_{j^o}},\theta_n\}$ that guarantees that the seller is indifferent between choosing $\theta_n$ and $\theta_{j^o}$. An additional condition, which we derive in the online appendix, ensures that he does not benefit from deviating to an alternative technology $\theta\notin\{\theta_{j^o},\theta_n\}$. Thus, an equilibrium with inefficient investment and bargaining delay exists under an open set of adoption costs when (\ref{costdifference2}) is satisfied. 

Finally, if $p_{\theta_{j^o}}<\theta_n$, then any equilibrium must be inefficient if $c_{j^o}\in\big(\frac{\theta_n-\theta_{j^o}}{2},\theta_n-\theta_{j^o}\big)$. This is because, if he invests efficiently, the seller's limiting equilibrium payoff is $p_{\theta_{j^o}}-\theta_{j^o}-c_{j^o}$. Given that $\theta_n>p_{\theta_{j^o}}$, he can deviate to $\theta_n$ and demand the entire surplus, which ensures a payoff that is arbitrarily close to $\frac{1-\theta_n}{2}$. This limiting value is strictly greater than $p_{\theta_{j^o}}-\theta_{j^o}-c_{j^o}$ whenever $c_{j^o}>\frac{\theta_n-\theta_{j^o}}{2}$.

\section{Concluding Remarks}\label{sec6}
We study a reputational bargaining model where a seller's production cost is determined \textit{endogenously} by whether he adopts an innovative technology before he bargains with a buyer. We show that in equilibrium there will be inefficient adoption and costly delays in reaching agreement, and that these inefficiencies arise if and only if there are large enough social gains from adopting the technology. Our analysis highlights the differences between models with exogenous distributions over production costs and ones where the distribution over production costs is endogenous. It also highlights 
the qualitative differences in the equilibrium outcomes when a player's private information is about their cost or value, compared to the case analyzed in APS  in which a player's private information is about their  discount rate. 
We conclude by discussing several extensions of our results once we vary our modeling assumptions.

\paragraph{The Timing of Offers:} In an earlier draft, we study the case where players make their initial offers \textit{simultaneously} and obtain similar results: In the game with an exogenous distribution over production costs, bargaining is efficient in all limiting equilibria when the difference between adjacent types' production costs is small, and efficient and inefficient limiting equilibria co-exist when the difference between adjacent types' production costs is large. The intuition is that when players make offers simultaneously, the seller does not know whether the buyer will make a screening offer or a pooling offer, which explains why both can arise in equilibrium. But in the game with endogenous technology adoption, there is a limiting equilibrium, in which the seller's adoption decision is socially inefficient if and only if the benefit from adoption is large enough.

Our main results partially extend to a model where the order with which players make offers is endogenous and, as in Kambe (1999), each player becomes committed with positive probability after making their initial offer. 
In this game, there \textit{exists} an equilibrium where the buyer makes an offer before the seller does, and players' equilibrium strategies coincide with those in the baseline model. Nevertheless, there also exist other equilibria due to the seller's incentive to signal his production cost. 
In particular, the off-path belief about the seller's cost has a significant effect on players' incentives when the seller can make an offer before the buyer does.

\paragraph{Bargaining Power:} In our model, players' bargaining powers are determined by the \textit{ratio} of their discount rates and a player has more bargaining power when they are more patient relative to their opponent. Our baseline model assumes that players share the same bargaining power. We discuss an extension of our results to the case where players have  
\textit{different discount rates}, which sheds light on the effects of bargaining power on investment incentives and delays.

We use $r_b$ to denote the buyer's discount rate and $r_s$ to denote the seller's discount rate. 
When the buyer's value for the object is $1$ and the seller's cost is $\theta$,
the equilibrium price in the Rubinstein bargaining game is $p_{\theta} \equiv \frac{r_b}{r_s+r_b} + \frac{r_s}{r_s+r_b} \theta$. 
Since the seller obtains a fraction $\frac{r_b}{r_b+r_s}$ of the total surplus $1-\theta$, his bargaining power is 
$\frac{r_b}{r_b+r_s}$ and the buyer's bargaining power is $\frac{r_s}{r_b+r_s}$.

Suppose first that $r_b/r_s$ is small enough so that $p_{\theta_1}<\theta_2$, which is the case when the buyer's bargaining power is relatively high. If $\pi(\theta_1)$ is above $\pi^*$, then the limiting equilibrium in the game with exogenous production costs is inefficient, in which the buyer offers $p_{\theta_1}$ and trades with delay. Otherwise, the limiting equilibrium is efficient. When the seller's adoption decision is endogenous, the adoption probability is bounded below one and there is significant delay in reaching agreement \textit{if and only if} $c\in\bigg(\frac{r_b}{r_b+r_s}(\theta_2-\theta_1),\theta_2-\theta_1\bigg)$. In particular, if $c<\theta_2-\theta_1$, in which case adopting the technology is socially optimal, and $r_b/r_s$ is arbitrarily small, the unique limiting equilibrium is inefficient except for an interval of adoption costs of vanishing Lebesgue measure.

As in the baseline model, the intuition behind the bargaining inefficiencies comes from the uninformed buyer's incentive to screen the informed seller. Screening is more attractive for the buyer when she has more bargaining power. By making a screening offer, the buyer will lower the price by approximately $\frac{r_s}{r_b+r_s}(\theta_2-\theta_1)$, which is a decreasing function of $r_b/r_s$. 

Conversely, in the case where $r_b/r_s$ is large enough so that $p_{\theta_1}>\theta_2$, the welfare properties of the limiting equilibrium will hinge on the size of $\theta_2-\theta_1$ in a similar way as in Theorems \ref{Theorem1} and \ref{Theorem2}. Specifically, if $\theta_2-\theta_1<\frac{r_b}{r_b+r_s}(1-\theta_2)$, then the unique limiting equilibrium is efficient, both under an exogenous distribution of production costs and under endogenous distributions of production costs. Otherwise, there may exist an inefficient equilibrium with under-adoption and delay for intermediate values of $c$, but the welfare loss vanishes as $r_b/r_s$ grows.

The above discussion highlights a stark contrast once we compare equilibrium welfare in the extreme cases where one of the player's discount rate is arbitrarily greater than their opponent's. In order to see this, let us focus on the case in which $c<\theta_2-\theta_1$. If the buyer is arbitrarily more patient than the seller, then in the unique limiting equilibrium, there is under-adoption and costly delay for almost all values of $c$. If the seller is arbitrarily more patient than the buyer, then every equilibrium is approximately efficient.

\end{spacing}
\appendix
\section{Lemma: An Implication of a Rich Set of Commitment Types}\label{sec4}
We establish a lemma which is implied by our richness assumption on the set of commitment types, that is, $1 \in \mathbf{P}_b$ and $\sup\mathbf{P}_s\setminus \{1\}=1$.
This lemma applies both in the case with an exogenous cost distribution and the case with endogenous technology adoption. 
Our conclusion in this lemma is extended to the case with multiple cost levels in the online appendix. 
For any $p_b\in \mathbf{P}_b$, let $\hat{\varepsilon}_b(p_b)$ 
denote the posterior probability with which the buyer is committed after she offers $p_b$. 
\begin{Lemma} \label{Lemma_4.5}
Fix any equilibrium. For any $p_b$ that satisfies $p_b \in \mathbf{P}_b\cap (\theta_1,\theta_2]$ and $\hat{\varepsilon}_b(p_b)<1$,
type-$\theta_2$ seller will demand $1$ with probability $1$ after observing the buyer offers $p_b$. 
\end{Lemma}
\begin{proof}
Suppose by way of contradiction that type $\theta_2$ offers $p_s<1$ with positive probability after the buyer offers $p_b \in \mathbf{P}_b\cap (\theta_1,\theta_2]$ with $\hat{\varepsilon}_b(p_b)<1$.
Our richness assumption that $\sup\mathbf{P}_s\setminus \{1\}=1$ implies that $ (p_s,1) \cap \mathbf{P}_s$ is non-empty. 
Since $p_b\leq \theta_2$, type $\theta_1$ must offer every $p_s'\in (p_s,1) \cap \mathbf{P}_s$ with positive probability.
This is because otherwise, the buyer's posterior assigns zero probability to type $\theta_1$ after observing $p_s'$ and will then concede immediately, in which case type $\theta_2$ has a strict incentive to deviate to $p_s'$.

After the seller offers $p_s$ and $p_s'$, let $T$ and $T'$ denote the times at which the rational-type buyer finishes conceding, let $c_b$ and $c_b'$ denote the buyer's concession probabilities at time $0$, and let $A$ and $A'$ denote the discounted probability of trade conditional on the seller never conceding and the buyer being the rational type. 
On the one hand, type $\theta_2$ weakly prefers $p_s$ to $p_s'$, which implies that
\begin{equation}\label{IC_theta2}
    (p_s-\theta_2)A\geq (p_s'-\theta_2)A'. 
\end{equation}
On the other hand, it is optimal for type-$\theta_1$ seller to offer $p_s'$ and to concede at time $T'$, so he prefers this strategy to offering $p_s$ and conceding at time $T$. This incentive constraint implies that:
\begin{equation*}
    e^{-rT'}\hat{\varepsilon}_b(p_b)(p_b-\theta_1)+(1-\hat{\varepsilon}_b(p_b))A'(p_s'-\theta_1)\geq e^{-rT}\hat{\varepsilon}_b(p_b)(p_b-\theta_1)+(1-\hat{\varepsilon}_b(p_b))A(p_s-\theta_1),
\end{equation*}
or
\begin{equation}\label{4.9}
 (e^{-rT'}-e^{-rT})\frac{\hat{\varepsilon}_b(p_b)}{1-\hat{\varepsilon}_b(p_b)}(p_b-\theta_1)\geq A(p_s-\theta_1)-A'(p'_s-\theta_1).
\end{equation}
We can bound the right-hand-side of (\ref{4.9}) using (\ref{IC_theta2}), which gives:
\begin{equation}\label{4.10}
   A(p_s-\theta_1)-A'(p'_s-\theta_1) =A(p_s-\theta_2)-A'(p_s'-\theta_2)+(A-A')(\theta_2-\theta_1)\geq (A-A')(\theta_2-\theta_1). 
\end{equation}

Since $p_s<p_s'$, inequality (\ref{IC_theta2}) implies that $A>A'$. According to (\ref{4.9}) and (\ref{4.10}), $A>A'$ implies that $T'<T$, which further implies that $T>0$. As a result, type $\theta_1$ must offer $p_s$ with positive probability. This is because otherwise, the buyer will concede immediately following $p_s$, which contradicts our earlier conclusion that $T>0$. Therefore, type $\theta_1$ must be indifferent between (i) offering $p_s$ and conceding at time $0$ and (ii) offering $p_s'$ and conceding at time $0$. This implies that $c_b(p_s-p_b)+p_b-\theta_1=c_b'(p_s'-p_b)+p_b-\theta_1$, or equivalently,
$c_b(p_s-p_b)=c_b'(p_s'-p_b)$.
Since $p_s'>p_s$, the fact that $c_b(p_s-p_b)=c_b'(p_s'-p_b)$ implies that $c_b \geq c_b'$.

The buyer's concession rate is $\lambda_b=\frac{r(p_b-\theta_1)}{p_s-p_b}$ when the seller offers $p_s$ and is $\lambda_b'=\frac{r(p_b-\theta_1)}{p'_s-p_b}$ after the seller offers $p_s'$. Since $p_s'>p_s$, we have $\lambda_b > \lambda_b'$. The expressions for $T$ and $T'$ imply that
\begin{align*}
    T=\frac{\log((1-c_b)/\hat{\varepsilon}_b(p_b))}{\lambda_b}\leq\frac{\log((1-c'_b)/\hat{\varepsilon}_b(p_b))}{\lambda'_b}=T'.
\end{align*}
This contradicts our earlier conclusion that $T'<T$.
\end{proof}

\section{Proof of Theorem 1}\label{sub4.1}
Our proof proceeds in five steps. In the first step, we characterize the unique equilibrium in the war-of-attrition game after players offer $(p_b,p_s) \in \mathbf{P}_b \times \mathbf{P}_s$. In the second step, we use players' continuation values in the war-of-attrition game to characterize the seller's equilibrium offer after observing the buyer's offer. In the third step, we use the buyer's sequential rationality constraint to show that in the limit, her offer must be close to either $p_{\theta_2}$ or $\min\{p_{\theta_1},\theta_2\}$, depending on the comparison between the probability of the low-cost type $\pi(\theta_1)$ and the cutoff $\pi^*$. In the fourth step, we characterize the expected delay in reaching an agreement. These four steps together imply that there is \textit{at most one} limiting equilibrium under generic $\pi$ and establish its properties stated in Theorem \ref{Theorem1}. The fifth step can be found in Online Appendix A, where we show that there exists at least one equilibrium.

We start from introducing some notation. Fix the cost distribution $\pi \in \Delta \{\theta_1,\theta_2\}$ and an equilibrium $(\sigma_b,\sigma_s,\tau_b,\tau_s)$. Let $\hat{\varepsilon}_b(p_b)$ and $\hat{\varepsilon}_s(p_b,p_s)$ denote the posterior probabilities with which the buyer and the seller, respectively, are commitment types after observing offers $(p_b,p_s)$. Let $\hat{\pi}_j(p_b,p_s)$ denote the probability that the seller is the rational type with production cost $\theta_j$ after observing offers $(p_b,p_s)\in \mathbf{P}_b\times \mathbf{P}_s$. According to Bayes' rule, the following equations hold in equilibrium:
\begin{align}
   & \hat{\varepsilon}_b(p_b) \equiv \frac{\varepsilon \mu_b(p_b)}{\varepsilon \mu_b(p_b)+(1-\varepsilon)\sigma_b(p_b)} \label{epsilonb} \\ & \hat{\varepsilon}_s(p_b,p_s) \equiv \frac{\varepsilon \mu_s(p_s)}{\varepsilon \mu_s(p_s)+(1-\varepsilon)[\pi\{\theta_1\} \sigma_s(p_s|\theta_1,p_b)+\pi\{\theta_2\}\sigma_s(p_s|\theta_2,p_b)]} \label{epsilons} \\
   & \hat{\pi}_j(p_b,p_s) \equiv \frac{(1-\varepsilon)\pi\{\theta_j\}\sigma_s(p_s|\theta_j,p_b)}{\varepsilon \mu_s(p_s)+(1-\varepsilon)[\pi\{\theta_1\} \sigma_s(p_s|\theta_1,p_b)+\pi\{\theta_2\}\sigma_s(p_s|\theta_2,p_b)]}. \label{pi}
\end{align}

\paragraph{Step 1.} In this step, we characterize the unique equilibrium in the war-of-attrition game after players offer $(p_b,p_s)\in \mathbf{P}_b\times \mathbf{P}_s$ that satisfy $1>p_s>p_b>\theta_1$.

Fix $(p_b,p_s)$ as well as the posterior beliefs about players' types $(\hat{\varepsilon}_b,\hat{\varepsilon}_s,\hat{\pi}_1,\hat{\pi}_2)$ after observing these offers. We denote the resulting continuation game by $\Gamma(p_b,p_s,\hat{\varepsilon}_b,\hat{\varepsilon}_s,\hat{\pi})$ and a pair of equilibrium strategies for the buyer and type-$\theta$ seller by $\tau_b\in\Delta(\mathbb{R}_+\cup \{+\infty\})$ and $\tau_s:\Theta\to \Delta(\mathbb{R}_+\cup \{+\infty\})$, respectively. 
Let 
$m\equiv \max\{j\in\{1,2\}: \theta_j<p_b\}$, 
which is well defined given that $p_b>\theta_1$. Recall that
$\lambda_s\equiv \frac{r(1-p_s)}{p_s-p_b}$
is the seller's concession rate that keeps the buyer indifferent between conceding and waiting. For every $j\in \{1,...,m\}$, recall that
$\lambda_b^j\equiv\frac{r(p_b-{\theta}_j)}{p_s-p_b}$
is the buyer's concession rate that keeps  type-${\theta}_j$ seller indifferent between conceding and waiting. Let $\lambda_b^{m+1}=\hat{\pi}_3=0$. 

If type-$\theta_j$ seller concedes with zero probability at time $0$ and concedes at rate $\lambda_s$ over the interval $(T^{j-1},T^j)$, with $0=T^0\leq T^1\leq ...\leq T^m$, then the probability with which the buyer's belief assigns to the event that \textit{the seller is either committed \textit{or} has a production cost that is strictly above $\theta_j$} will reach $1$ at time
\begin{equation}
    T_s^j\equiv \frac{-\log(\hat{\varepsilon}_s+\sum_{i>j}\hat{\pi}_i)}{\lambda_s} \label{T_s}.
\end{equation}
Likewise, if the buyer concedes with zero probability at time $0$ and concedes at rate $\lambda_b^j$ over the time interval $(T^{j-1},T^j)$, then the buyer finishes conceding at time
\begin{align*}
    T_b\equiv\frac{-\log(\hat{\varepsilon}_b)-\sum_{j=1}^{m-1} (\lambda_b^j-\lambda_b^{j+1})T^j}{\lambda_b^{m}}.
\end{align*}
In equilibrium, both players must finish conceding at the same time. Therefore, one of them will concede with strictly positive probability at time $0$ as long as $T_b\neq T_s^m$. Let
\begin{equation} \label{Lratio}
    L\equiv \frac{-\lambda_s\log\hat{\varepsilon}_b}{-\sum\limits_{j=1}^m (\lambda_b^j-\lambda_b^{j+1})\log(\hat{\varepsilon}_s+\hat{\pi}_{j+1})}.
\end{equation}
One can verify that $L<1$ if and only if $T_b<T_s^m$. Therefore, the seller concedes with strictly positive probability at time $0$ if and only if $L<1$ and the buyer concedes with strictly positive probability at time $0$ if and only if $L>1$. We refer to the player who concedes at time $0$ with strictly positive probability as the \emph{weak player}. In order to derive the probability with which the weak player concedes to their opponent at time $0$, let
\begin{equation*}
    \hat{c}_s^i \equiv 1-\Bigg(\hat{\varepsilon}_b^{-\lambda_s}\prod\limits_{j=i}^m (\hat{\varepsilon}_s+\hat{\pi}_{j+1})^{\lambda_b^j-\lambda_b^{j+1}}\Bigg)^{1/\lambda_b^i} \quad \textrm{for every } i \in \{1,...,m\}
\end{equation*}
\begin{equation*}
\hat{c}_b=1-\hat{\varepsilon}_b e^{\sum\limits_{j=1}^m \lambda_b^j(T^j_s-T^{j-1}_s)} .
\end{equation*}
Let 
$j^*\equiv \min\{j\in\{1,...,m\}:\hat{c}^j_s<\sum_{i\leq j}\hat{\pi}_i\}$. 
Suppose the buyer is the weak player. Then $\hat{c}_b$ is the probability that the buyer concedes at time $0$ so that the rational-type buyer finishes conceding at time $T_s^m$. Likewise, if the seller is the weak player and $j^*=1$, then $\hat{c}_s^1$ is the probability with which type $\theta_1$ concedes at time $0$ so that the buyer's belief that the seller is either committed or that $\theta\geq p_b$ reaches $1$ at time $T_b$. However, if $j^*>1$, it is not sufficient to have type $\theta_1$ conceding with probability $1$ at time $0$ to make both players finish conceding at the same time. Instead, we need all types strictly below $\theta_{j^*}$ to concede at time $0$ with probability $1$ and possibly type $\theta_{j^*}$ to concede at time $0$ with positive probability. As a result, when the seller is the weak player, his concession probability at time $0$ equals $\hat{c}_s^{j*}$. 
Lemma \ref{L.A1} summarizes these findings:
\begin{Lemma}\label{L.A1}
Fix any pair of offers $(p_b,p_s)\in\mathbf{P}_b\times \mathbf{P}_s$ that satisfy $1>p_s>p_b>\theta_1$. In any equilibrium of $\Gamma(p_b,p_s,\hat{\varepsilon}_b,\hat{\varepsilon}_s,\hat{\pi})$, the buyer concedes with positive probability at time zero if and only if $L>1$ and the seller concedes with positive probability at time $0$ if and only if $L<1$. Players' concession probabilities at time $0$ are $c_b\equiv\max\{0,\hat{c}_b\}$ and $c_s\equiv\max\{0,\hat{c}_s^{j^*}\}$, respectively. 
\end{Lemma}
A formal proof of Lemma \ref{L.A1} can be found in APS, which we omit in order to avoid repetition.

Next, consider the continuation game in which no player concedes at time $0$. In equilibrium, for every
$j\in\{j^*,...,m-1\}$, 
type ${\theta}_j$ will finish conceding at time
\begin{equation} \label{Tdef}
    T^j= T_s^j+\frac{\log(1-c_s)}{\lambda_s}. 
\end{equation}
In addition, the rational types of both players will finish conceding at the same time
\begin{equation}
    T^m\equiv \min\left\{\frac{-\log(\hat{\varepsilon}_b)-\sum_{j=j^*}^{m-1} (\lambda_b^j-\lambda_b^{j+1})T^j_s}{\lambda_b^m},T_s^m\right\}.
\end{equation}
Lemma \ref{L.A2} characterizes players' equilibrium strategies in the war-of-attrition game:
\begin{Lemma}\label{L.A2}
In every equilibrium of the war-of-attrition game $\Gamma(p_b,p_s,\hat{\varepsilon}_b,\hat{\varepsilon}_s,\hat{\pi})$ with $(p_b,p_s)\in\mathbf{P}_b\times\mathbf{P}_s$ and $1>p_s>p_b>\theta_1$, the buyer's and the seller's concession times $\tau_b$ and $\tau_s(\theta)$  satisfy:
\begin{enumerate}
    \item For every $j \in \{j^*,...,m\}$, the buyer concedes at rate $\lambda_b^j$
    when $t \in (T^{j-1},T^j)$ with $T^{j^*-1}=0$.
    \item The seller with cost $\theta\in\{{\theta}_{j^*},...,{\theta}_m\}$ concedes at rate $\lambda_s$ when $t \in (T^{j-1},T^j)$ with $T^{j^*-1}=0$. 
    \item The seller never concedes if his production cost is strictly greater than ${\theta}_m$.
\end{enumerate}
\end{Lemma}
Next, we characterize players' concession probabilities at time $0$ in the limit where $\varepsilon \rightarrow 0$. 
Formally, consider an infinite sequence of strictly positive numbers $\{\varepsilon^k\}_{k=0}^{+\infty}$ that satisfy
 $\lim_{k \rightarrow +\infty}\varepsilon^k= 0$. Let $(\sigma_b^k,\sigma_s^k)$ be players' equilibrium strategies in the war-of-attrition game when the probability of commitment type is $\varepsilon^k$. Without loss of generality, we focus on the case where $(\sigma_b^k,\sigma_s^k)$ converges to $(\sigma_b^\infty,\sigma_s^\infty)$.\footnote{This is because otherwise, we can apply the Helly's selection theorem (Billingsley, 2013b), that $\Delta[0,1]$ is sequentially compact in the topology of weak convergence, and find a converging subsequence and focus on that subsequence.} 
Let $(\hat{\varepsilon}^k_b,\hat{\varepsilon}^k_s,\hat{\pi}^k)$ be given by (\ref{epsilonb}), (\ref{epsilons}) and (\ref{pi}) using $(\varepsilon^k,\sigma_b^k,\sigma_s^k)$, and let $\lim_{k\to \infty}\hat{\pi}_j^k=\hat{\pi}_j^\infty$ for every $j \in \{1,2\}$ and $\hat{\varepsilon}_i^\infty \equiv \lim_{k\to \infty}\hat{\varepsilon}_i^k$ for every $i \in \{b,s\}$.

\begin{Lemma} \label{L.A3}
Suppose $\{\varepsilon^k\}_{k=1}^\infty$ is such that $\varepsilon^k\to 0$ as $k\to \infty$. Let $(c_b^k,c_s^k)_{k=1}^\infty$ be given according to Lemma \ref{L.A1} in the game $\Gamma(p_b,p_s,\hat{\varepsilon}^k_b,\hat{\varepsilon}^k_s,\hat{\pi}^k)$ with $(p_b,p_s)\in\mathbf{P}_b\times\mathbf{P}_s$ and $1>p_s>p_b>\theta_1$, and let $(c_b^\infty,c_s^\infty)$ be the limit of this sequence as $k\to \infty$. Then,
\begin{itemize}
    \item [1.] If $\sigma^\infty_s(p_s|p_b,\theta_1)+\sigma^\infty_s(p_s|p_b,\theta_2)>0$ and $p_b>\theta_2$, and if $\hat{\varepsilon}^\infty_b(p_b)>0$ or $\lambda_b^2>\lambda_s$, then $c_s^\infty(p_b,p_s)=1$.
    \item [2.] If $\sigma^\infty_b(p_b)>0$ and $\hat{\pi}^\infty_2(p_b,p_s)>0$, and if $\lambda_s>\lambda_b^2$ or $p_b\leq \theta_2$, then $c_b^\infty(p_b,p_s)=1$.  
 \item [3.] If $\sigma^\infty_b(p_b)>0$, and if $\hat{\varepsilon}_s^\infty(p_b,p_s)>0$ or $\lambda_s>\lambda_b^1$, then $c_b^\infty(p_b,p_s)= 1$.
    \item [4.] If $\sigma^\infty_s(p_s|p_b,\theta_1)>0$ and $\hat{\pi}_2^\infty(p_b,p_s)=0$, and if $\hat{\varepsilon}^\infty_b(p_b)>0$, or if $\lim_{k\to \infty}(\hat{\varepsilon}_b^k(p_b)/\hat{\pi}_2^k(p_b,p_s))>0$ and $\lambda_b^1>\lambda_s$, then $c_s^\infty(p_b,p_s)=1$.
\end{itemize}
\end{Lemma}

% One consequence of Lemma \ref{L.A3} is that as long as the buyer assigns strictly positive probability to
% type $\theta_2$, the identity of the player who concedes at time $0$ is determined by the comparison of concession rates between the buyer and the type-$\theta_2$ seller. 

\paragraph{Step 2.} In this step, we use Lemmas \ref{L.A1}, \ref{L.A2}, and \ref{L.A3} 
to derive the seller's equilibrium offers when $\varepsilon$ and $\nu$ are close to $0$. For $j \in \{1,2\}$, let $p_j(p_b) \equiv \max\{p\in\mathbf{P}_s: p\leq 1+\theta_j-p_b\}$. The compactness of $\mathbf{P}_s$ ensures that this is well-defined. First, we show that when the buyer offers a price $p_b\in (\theta_2,p_{\theta_2}]$ with a probability bounded above $0$,  
both types of the seller will offer the same price that is close to $p_2(p_b)$ in every equilibrium. 
\begin{Lemma} \label{lemma5}
For every $\eta>0$ and $\nu>0$, there exists $\bar{\varepsilon}_{\nu} >0$ such that for every $\theta\in\Theta$, every $\varepsilon\in(0,\varepsilon_\nu)$, and every $p_b\in(\theta_2,p_{\theta_2}]$ such that $\sigma_b(p_b)>\eta$, $\sigma_s(\cdot|p_b,\theta)$ is $\eta$-close to the Dirac measure on $p_2(p_b)$. 
\end{Lemma}
\begin{proof}

Fix $\eta>0$ and a small enough $\nu>0$. Suppose by way of contradiction that for all $\overline{\varepsilon}>0$, there exist $p_b$ and $\varepsilon$ with $\sigma_b(p_b)>\eta$ and $\varepsilon\in (0,\overline{\varepsilon})$ such that $\sigma_{s}(p_s'|p_b,\theta_2)>\eta$ for some $p'_s<p_2(p_b)$. Without loss of generality, we can take $\overline{\varepsilon}$ to be arbitrarily small. We argue that it must be the case that type $\theta_1$ offers $p_2(p_b)$ with probability bounded above zero and that type $\theta_2$ does so with vanishing probability. Suppose by way of contradiction that this is not the case. According to Parts 2 and 3 of Lemma \ref{L.A3}, for every $\delta>0$, there is $\overline{\varepsilon}_\delta>0$ such that $\varepsilon<\overline{\varepsilon}_\delta$ implies that the buyer's time-zero concession probability after the seller offers $p_2(p_b)$ is at least $1-\delta$. Thus, type $\theta_2$'s payoff when he offers $p_2(p_b)$ is at least $(1-\delta)(p_2(p_b)-\theta_2)$. Since $\delta>0$ is arbitrary, we can take it to be sufficiently small, in which case we get that, for $\varepsilon$ sufficiently small, the high type's payoff when offering $p_2(p_b)$ is strictly higher than the upper bound on his equilibrium payoff $p_s'-\theta_2$. This leads to a contradiction when $\overline{\varepsilon}\leq\overline{\varepsilon}_\delta$. 

 On the other hand, by Part 2 of Lemma \ref{L.A3}, $p_s'<p_2(p_b)$ implies that for every $\delta>0$, there exists $\overline{\varepsilon}_\delta$ such that $\varepsilon<\overline{\varepsilon}_\delta$ implies that the buyer's time zero concession probability after the seller offers $p_s'$ is at least $1-\delta$. Let $T_1$ be the time at which the rational buyer finishes conceding against the low-type seller (as defined in \eqref{Tdef}) after the seller offers $p_2(p_b)$, and let $c_b$ be the associated time-zero concession probability for the buyer. Note that $T_1\to+\infty$ as $\varepsilon\to 0$. For $\varepsilon$ sufficiently small, the low type's incentive to offer $p_2(p_b)$ instead of $p_s'$ implies that
$c_b(p_2(p_b)-\theta_1)+(1-c_b)(p_b-\theta_1)\geq p_s'-\theta_1$.
Moreover, type $\theta_2$'s payoff from deviating to $p_2(p_b)$ and waiting until $T_1$ for the buyer to concede is at least
  \begin{equation*}
      \bigg(c_b+(1-c_b)\frac{p_b-\theta_1}{p_2(p_b)-\theta_1}\bigg)(p_2(p_b)-\theta_2)\geq \frac{p_s'-\theta_1}{p_2(p_b)-\theta_1}(p_2(p_b)-\theta_2)>p_s'-\theta_2.
  \end{equation*}
Thus, for $\overline{\varepsilon}$ sufficiently small and any $\varepsilon<\overline{\varepsilon}$, type $\theta_2$ strictly benefits from deviating to $p_2(p_b)$, which is a contradiction.  

If for all $\overline{\varepsilon}>0$, there exists $p_b$ and $\varepsilon$ with $\sigma_b(p_b)>\eta$ and $\varepsilon\in (0,\overline{\varepsilon})$ such that $\sigma_{s}(p_s'|p_b,\theta_2)>\eta$ for some $p'_s>p_2(p_b)$, then Part 1 of Lemma \ref{L.A3} implies that type $\theta_2$ concedes with positive probability at time zero for $\varepsilon$ sufficiently small, and therefore his equilibrium payoff is $p_b-\theta_2$. An analogous argument to the one in the previous paragraph then implies that for a sufficiently small $\varepsilon$, type $\theta_2$ obtains a weakly higher payoff from deviating to $p_2(p_b)$ (strictly so unless $p_b=p_2(p_b)=p_{\theta_2}$). Thus, type $\theta_2$'s strategy is $\eta$-close to the Dirac measure on $p_2(p_b)$. Given this, Part 2 of Lemma \ref{L.A3} implies that for all $\delta>0$ there exists $\overline{\varepsilon}_\delta>0$ such that for any $\varepsilon\in(0,\overline{\varepsilon}_\delta)$, type $\theta_1$'s payoff from offering $p_2(p_b)$ is at least $p_2(p_b)-\theta_1-\delta$. If he offers $p_s>p_2(p_b)$ with probability greater than $\eta$, then for sufficiently small $\varepsilon$, type $\theta_1$ receives a payoff of $p_b-\theta_1$, which is strictly dominated by offering $p_2(p_b)$ whenever $p_b<p_2(p_b)$. Thus, type $\theta_1$'s offer converges to the Dirac measure on $p_2(p_b)$ as well. 
\end{proof}

Next, we characterize the low-cost seller's offer conditional on the buyer offering $p_b\leq \min\{p_{\theta_1},\theta_2\}$.

\begin{Lemma} \label{lemma4.6}
For every $\eta>0$ and $\nu>0$, there exists $\bar{\varepsilon}_{\nu} >0$ such that for every $\varepsilon\in(0,\varepsilon_\nu)$, and every $p_b\in(\theta_1,\min\{p_{\theta_1},{\theta_2}\}]$ such that $\sigma_b(p_b)>\eta$, $\sigma_s(\cdot|p_b,\theta_1)$ is $\eta$-close to the Dirac measure on $p_1(p_b)$. 
\end{Lemma}

\begin{proof}
 Fix $\eta>0$ and an arbitrarily small $\nu>0$. Let $p_b\in \mathbf{P}_b\cap (\theta_1,\min\{p_{\theta_1},\theta_2\}]$ with $\sigma_b(p_b)>\eta$. First, we show that the low type has to concede immediately if he demands $1$ with positive probability. To see this, let $T_1 \geq 0$ denote the last instant at which the low type concedes after demanding $1$. Such $T_1$ is well-defined since all rational types finish conceding in finite time. The buyer's payoff from conceding is $0$, and therefore, she has a strict incentive to wait until $T_1$ before conceding. If $T_1>0$, the fact that $p_b>\theta_1$ and that the buyer does not concede before $T_1$ implies that the low type has a strict incentive to concede at time $0$. This implies that $T_1=0$, so the low type must concede at time $0$ with probability $1$ after demanding $1$, which by our convention in Section \ref{sec2}, is equivalent to making a counteroffer of $p_b$.  

Second, we consider the case in which the low type offers $p_s<1$ after the buyer offers $p_b$. For $\nu$ small enough, $p_b\leq p_{\theta_1}$ implies that $p_b\leq p_1(p_b)$. Part 3 of Lemma \ref{L.A3} implies that for all $\delta>0$ there exists $\overline{\varepsilon}_\delta$ such that $\varepsilon<\overline{\varepsilon}_\delta$ implies that type $\theta_1$ can induce the buyer to concede with probability greater than $1-\delta$ by offering $p_1(p_b)$. If type $\theta_1$ offers anything greater than $p_1(p_b)$ with probability greater than $\eta$, then by Lemma \ref{Lemma_4.5} and Part 4 of Lemma \ref{L.A3}, there exists $\overline{\varepsilon}_\eta>0$ such that when $\varepsilon\in(0,\overline{\varepsilon}_\eta)$, he will concede at time $0$ with positive probability, which is weakly dominated by offering $p_1(p_b)$, and it is strictly dominated by offering $p_1(p_b)$ unless $p_b=p_1(p_b)$. Therefore, type $\theta_1$ will offer $p_1(p_b)$ with probability converging to $1$ as $\varepsilon\to 0$.  
\end{proof}

\paragraph{Step 3.} In this step, we characterize the buyer's equilibrium offer. For every $\eta>0$, there exists $\gamma>0$ such that when $0<\nu,\varepsilon<\gamma$, the buyer's payoff from offering any $p_b$ weakly greater than $p_{\theta_2}$ is no less than $1-p_b-\eta$. This is because the seller's concession rate is strictly less than the buyer's concession rate if he demands 
anything greater than $p_b$ and, by Part 1 of Lemma \ref{L.A3}, the seller's probability of concession at time $0$ converges to $1$ 
as $\varepsilon \rightarrow 0$. Combining this with Lemma \ref{lemma5}, we know that for every $\eta>0$ there exists $\overline{\nu}>0$ such that when $\nu<\overline{\nu}$ there exists $\overline{\varepsilon}_\nu$ such that when $\varepsilon\in(0,\overline{\varepsilon}_\nu)$, any offer $p_b\in\mathbf{P}_b$ that satisfies $p_b>\theta_2$ and $\sigma_b(p_b)>\eta$ yields the buyer a payoff that is at least $1-\max\{p_b,p_2(p_b)\}-\eta$.

Next, we derive the buyer's payoff after she offers a price $p_b$ that satisfies $p_b\in\mathbf{P}_b\cap(\theta_1,\theta_2]$. First, suppose the buyer offers $p_b$ with zero probability in equilibrium. Then, the seller's posterior belief assigns probability $1$ to the commitment type after observing $p_b$, after which type $\theta_1$ will concede immediately and the buyer's payoff is at least $\pi(\theta_1)(1-p_b)$. Second, suppose the buyer offers $p_b$ with positive probability in equilibrium, in which case Lemma \ref{Lemma_4.5} implies that type-$\theta_2$ seller will counteroffer $1$ with probability $1$. If $p_b>p_{\theta_1}$, then Part 4 of Lemma \ref{L.A3} implies that taking ${\varepsilon}$ to be sufficiently small, the low type has to concede with probability greater than $1-\eta$ following any offer in the support of his equilibrium strategy, and thus the buyer's payoff is at least $\pi(\theta_1)(1-p_b)-\eta$. If $p_b\leq p_{\theta_1}$ and $\sigma_b^\infty(p_b)>0$, Lemma \ref{lemma4.6} says that for $\varepsilon$ sufficiently small, the low type's offer is $p_1(p_b)$ with probability greater than $1-\eta$. If $p_b\leq p_{\theta_1}$ and $\sigma_b^\infty(p_b)=0$, then Part 4 of Lemma \ref{L.A3} implies that the low type has to concede with probability greater than $1-\eta$ if he offers anything greater than $p_1(p_b)$ with probability greater than $\eta$, and therefore the buyer's payoff from offering $p_b$ in this case is at least $\pi(\theta_1)(1-p_1(p_b))-\eta$. Combining these cases, for $\varepsilon$ sufficiently small, the buyer's payoff when she offers $p_b\in\mathbf{P}_b\cap(\theta_1,\theta_2]$ is bounded below by $\pi(\theta_1)(1-\max\{p_1(p_b),p_b\})$ and the lower bound is attained in the limit as $\varepsilon\to 0$ if the buyer offers $p_b$ with non-vanishing probability in equilibrium. 

Let $p_1^\nu=\argmax\limits_{p\in \mathbf{P}_b\cap [\theta_1,\theta_2]} (1-\max\{p_1(p_b),p_b\})$ and $p_2^\nu=\argmax\limits_{p\in \mathbf{P}_b\cap [\theta_2,1]} (1-\max\{p_2(p_b),p_b\})$. Our earlier arguments imply that for a fixed $\nu>0$, as $\varepsilon$ goes to zero, $\sigma_b$ must be arbitrarily close to a distribution belonging to $\Delta(p_1^\nu\cup p_2^\nu)$. Moreover, for $\nu>0$ sufficiently small, if $\pi(\theta_1)<\pi^*$, the buyer strictly prefers offering $p\in p_2^\nu$ for a limiting payoff of $1-p_{\theta_2}$ over offering $p\in p_1^\nu$ for a limiting payoff of $\pi(\theta_1)(\min\{p_{\theta_1},\theta_2\}-\theta_1)$, and thus limiting equilibrium strategies in this case are characterized by the first part in Theorem \ref{Theorem1}. If $\pi(\theta_1)>\pi^*$, the buyer strictly prefers to offer $p\in p_1^\nu$. Therefore, the equilibrium strategies must converge to those described in part two of Theorem \ref{Theorem1}. 

\paragraph{Step 4.} In this step, we compute the expected delay in reaching an agreement. If $\pi(\theta_1)<\pi^*$, by the above arguments, the resulting equilibrium outcome is approximately efficient, with an agreement being reached with negligible delay at a price arbitrarily close to $p_{\theta_2}$.  

If $\pi(\theta_1)>\pi^*$ and $p_{\theta_1}< \theta_2$, then in the limit, the buyer offers $p_{\theta_1}$ and type $\theta_1$ accepts. Otherwise, the buyer offers $\theta_2$ in equilibrium, after which type $\theta_1$ raises the price to $p_1(\theta_2)>\theta_2$ after which players play the war-of-attrition game. However, Part 3 of Lemma \ref{L.A3} shows that the expected delay in the resulting war-of-attrition vanishes as $\varepsilon \rightarrow 0$ and thus the equilibrium outcome, conditional on the seller's cost is $\theta_1$, is approximately efficient. If the seller's cost is $\theta_2$, he will respond to the buyer's equilibrium offer by demanding the entire surplus $1$. In order to deter a deviation from the low type, the buyer must wait a considerable amount of time before conceding to this demand. As argued in Section \ref{sub3.1}, the incentive constraints of the seller pin down the limiting expected welfare loss from delay to be given by (\ref{expecteddelay}). To complete the argument provided there, we show how to derive (\ref{limitinglowerbound}) from (\ref{lowerbound}). In order to avoid the buyer from conceding immediately after the seller demands $p_s\in\mathbf{P}_s$ with $p_s\in(1-\varepsilon,1)$, which would in turn give rise to a profitable deviation for the seller, it must be that type $\theta_1$ demands $p_s$ with positive probability (by Lemma \ref{L.A3}). As a result, type $\theta_1$'s incentive constraint requires that $c_b p_s+(1-c_b)\min\{p_{\theta_1},\theta_2\}-\theta_1\geq1-\min\{p_{\theta_1},\theta_2\}$. Plugging this into (\ref{lowerbound}) and using the fact that $T_1\to +\infty$ as $\varepsilon\to 0$, we obtain (\ref{limitinglowerbound}).

%Therefore, every equilibrium must satisfy the properties stated in Theorem \ref{Theorem1}. We establish the existence of equilibrium in Online Appendix A, which completes the proof of Theorem \ref{Theorem1}.

\section{Proof of Theorem 2}\label{sub4.2}
In this appendix, we establish some necessary conditions for equilibrium. We establish the existence of equilibrium in Online Appendix B. Throughout the proof, we use $V_\theta$ to denote the equilibrium payoff of type $\theta$ at the bargaining stage (net of adoption costs) and we use $\pi(\theta_1)$ to denote the seller's equilibrium adoption probability. The following series of lemmas provide necessary conditions for the limiting equilibria under every parameter configuration, which, together with the existence result in Online Appendix B, establishes Theorem \ref{Theorem2}.

First, consider the case in which $c>\theta_2-\theta_1$, that is, the cost of adoption is strictly greater than the social benefit from adoption. We show that the seller adopts with zero probability in every equilibrium. This in turn implies that the buyer  has no incentive to offer anything that is strictly lower than $\theta_2$. As a result, her offer in equilibrium converges to $p_{\theta_2} \equiv \frac{1+\theta_2}{2}$ as $\varepsilon \rightarrow 0$. 
\begin{Lemma} \label{lemmab2}
   If $c>\theta_2-\theta_1$, then for every $\eta>0$, there exists $\bar{\nu}>0$ such that when $\nu<\bar{\nu}$, there exists $\bar{\varepsilon}_{\nu} >0$ such that for every $\varepsilon \in (0,\bar{\varepsilon}_{\nu})$, the adoption probability is $0$ and the expected delay is less than $\eta$ in every equilibrium.
\end{Lemma}
\begin{proof}
 Suppose by way of contradiction that the seller adopts with strictly positive probability. His payoff in the bargaining stage after he adopts is $V_{\theta_1} \equiv \mathbb{E}[e^{-r\tau}(p-\theta_1)|\theta=\theta_1]$, where $\tau$ is the time of trade and $p$ is the trading price. If the seller deviates to not adopting and uses type $\theta_1$'s strategy in the war-of-attrition game, then he can secure a payoff of $\mathbb{E}[e^{-r\tau}(p-\theta_2)|\theta=\theta_1]$. As a result, 
 \begin{equation*}
 V_{\theta_1}-c-V_{\theta_2}\leq \mathbb{E}[e^{-r\tau}|\theta=\theta_1](\theta_2-\theta_1)-c\leq (\theta_2-\theta_1)-c<0,
 \end{equation*}
which implies that the seller strictly prefers not to adopt. This leads to a contradiction. 

Given that the seller adopts with zero probability, Proposition 3 in Abreu and Gul (2000) implies that in the limit, in any equilibrium, players trade at a price of $p_{\theta_2}$ with negligible delay.
\end{proof}

The rest of this proof considers the case in which $c<\theta_2-\theta_1$. We start from the subcase in which $p_{\theta_1}<\theta_2$ and $c\in\big(\frac{\theta_2-\theta_1}{2},\theta_2-\theta_1\big)$. The first condition ensures that the gap between $\theta_2$ and $\theta_1$ is large enough so that the buyer benefits from screening under certain values of $\pi(\theta_1)$. The second condition implies that the adoption cost is neither too high nor too low, which ensures that the seller is willing to mix at the adoption stage with probabilities that make the buyer indifferent between the screening offer $p_{\theta_1}$ and the pooling offer $p_{\theta_2}$. The buyer's mixing probabilities over $p_{\theta_1}$ and $p_{\theta_2}$ are chosen in order to make the seller indifferent at the adoption stage. Lemma \ref{lemmab3} characterizes the unique limiting equilibrium outcome under these two conditions. 
\begin{Lemma} \label{lemmab3}
Suppose $p_{\theta_1}<\theta_2$ and $c\in \bigg(\frac{\theta_2-\theta_1}{2},\theta_2-\theta_1\bigg)$. For every $\eta>0$, there exists $\bar{\nu}>0$ such that when $\nu<\bar{\nu}$, there exists $\bar{\varepsilon}_{\nu} >0$ such that for every $\varepsilon \in (0,\bar{\varepsilon}_{\nu})$, the adoption probability is $\eta$-close to $\pi^*$ and the expected delay is bounded above $0$ in every equilibrium.
\end{Lemma}
\begin{proof}
Fix any $\nu>0$ that is small enough. Suppose by way of contradiction that there exists $\eta>0$ such that for all $\overline{\varepsilon}>0$ there exists $\varepsilon\in(0,\overline{\varepsilon})$ such that $|\pi(\theta_1)-\pi^*|>\eta$. Without loss, take $\overline{\varepsilon}$ to be arbitrarily small. If $\pi(\theta_1)>\pi^*+\eta$, we show that the buyer's offer is arbitrarily close to $p_{\theta_1}$ when $\varepsilon\in(0,\overline{\varepsilon})$. According to Lemma \ref{Lemma_4.5}, if the buyer offers $p_b\leq \theta_2$ with positive probability, then type $\theta_2$ will offer $1$. As we showed in Step 3 of the proof of Theorem \ref{Theorem1}, for every $\delta>0$ there exists $\overline{\varepsilon}_\delta>0$ such that $\varepsilon\in(0,\overline{\varepsilon}_\delta)$ implies that the buyer can secure a payoff of at least $\pi(\theta_1)(1-\max\{p_b,p_1(p_b)\})-\delta$ by offering $p_b\leq \theta_2$ and that the bound is attained for any $p_b$ satisfying $\sigma_b(p_b)>\delta$, which for $\delta$ sufficiently small, is maximized by choosing $p^*_b\in \mathbf{P}_b\cap (p_{\theta_1}-\nu,p_{\theta_1}+\nu)$. If $1-\eta>\pi(\theta_1)>\pi^*+\eta$ and $\nu$ is sufficiently small, the resulting payoff is strictly greater than the buyer's highest payoff when she offers $p_b>\theta_2$, which is arbitrarily close to $1-p_{\theta_2}$. If $\pi(\theta_1)>1-\eta$ for $\eta$ arbitrarily small, the payoff from offering $p_b>\theta_2$ is at most $1-p_b<\pi(\theta_1)(1-p_{\theta_1})$. Therefore, for every $\eta>0$, there exist $\nu$ and $\varepsilon$ close enough to $0$ so that in every equilibrium with $\pi(\theta_1)>\pi^*+\eta$, the buyer's offer will belong to $\mathbf{P}_b\cap (p_{\theta_1}-\eta,p_{\theta_1}+\eta)$ with probability more than $1-\eta$, and therefore the seller's equilibrium payoff conditional on adopting is arbitrarily close to $p_{\theta_1}-\theta_1$. 
    
Moreover, for each $\delta>0$ there exists $\overline{\varepsilon}_\delta>0$ such that if $\varepsilon\in(0,\overline{\varepsilon}_\delta)$, then by not adopting the technology, the seller can secure a payoff of at least $\frac{1-\theta_2}{2}-\frac{1-\theta_2}{1-\theta_1}\nu-\delta$. To show this, suppose that the seller does not adopt, after which he offers $p_s\in\mathbf{P}_s\cap (1-{\varepsilon},1)$ following any $p_b^*$ such that $\sigma_b(p_b^*)>\eta$. By Parts 2 and 3 of Lemma \ref{L.A3}, there exists $\overline{\varepsilon}_\delta>0$ such that if $\varepsilon\in(0,\overline{\varepsilon}_\delta)$ the buyer will concede at time $0$ with probability greater than $1-\delta$ after the  seller offers $p_s$ unless type $\theta_1$ offers $p_s$ with positive probability. 
Type $\theta_1$'s incentive to offer $p_s$ after the buyer offers $p_b^*$ implies that
    \begin{equation*}
    p_b^*-\theta_1+c_b^*(\max\{p^*_b,p_1(p^*_b)\}-p_b^*)= e^{-rT}\hat{\varepsilon}_b(p_b^*)(p_b^*-\theta_1)+(1-\hat{\varepsilon}_b(p_b^*))A(p_s-\theta_1),
    \end{equation*}
    where $c_b^*$ is the 
    probability that the buyer concedes at time $0$
    after the seller offers $\max\{p^*_b,p_1(p^*_b)\}$,\footnote{If $\max\{p^*_b,p_1(p^*_b)\}=p_b^*$, then we can set $c_b^*=1$.} 
    and $A$ and $T$ are the discounted concession probabilities of the rational-type buyer and the time at which the rational-type buyer finishes conceding following the seller's offer $p_s$. 
    
Therefore, taking the limit as $p_s\to 1$, our above arguments imply that for every $\delta>0$, there exists $\overline{\varepsilon}_\delta>0$ such that $\varepsilon<\overline{\varepsilon}_\delta$ implies that the seller's payoff from not adopting the technology is at least $\frac{p_{\theta_1}-\nu-\theta_1}{1-\theta_1}(1-\theta_2)-\delta$. By setting $\nu$ and $\delta$ to be sufficiently low, 
    this payoff lower bound is strictly greater than $p_{\theta_1}-\theta_1-c$ whenever $c>\frac{\theta_2-\theta_1}{2}$. This contradicts $\pi(\theta_1)>0$. 

Next, suppose by way of contradiction that $\pi(\theta_1)<\pi^*-\eta$ for some sufficiently small $\varepsilon$. For every $\delta>0$ there exists $\overline{\nu}$ such that for all $\nu<\overline{\nu}$ there exists $\overline{\varepsilon}_\nu$ such that when $\varepsilon\in(0,\overline{\varepsilon}_\nu)$, the buyer's incentive constraint requires that any price that she offers with probability greater than $\eta$ satisfies $p_b^*\in(p_{\theta_2}-\eta,p_{\theta_2}+\eta)$, and therefore $V_{\theta_1}-c$ is within $\eta$-distance from $p_{\theta_2}-\theta_1-c$, and $V_{\theta_2}$ is within $\eta$-distance from $p_{\theta_2}-\theta_2$. Hence, for sufficiently small $\nu>0$ and $\varepsilon>0$, we obtain that $V_{\theta_1}-c>V_{\theta_2}$, where the inequality follows from $c<\theta_2-\theta_1$. This contradicts $\pi(\theta_1)<1$.

Therefore, in any equilibrium, the seller must adopt with probability close to $\pi^*$. As in the proof of Theorem \ref{Theorem1}, for every $\delta>0$ there exists $\overline{\varepsilon}_\delta>0$ such that if $\varepsilon\in(0,\overline{\varepsilon}_\delta)$, then the buyer's sequential rationality constraint requires that any price that she offers with probability greater than $\eta$ must belong to $(p_{\theta_1}-\nu,p_{\theta_1}+\nu)\cup (p_{\theta_2}-\nu,p_{\theta_2}+\nu)$. After the buyer offers $p\in (p_{\theta_2}-\nu,p_{\theta_2}+\nu)$, trade happens with delay smaller than $\eta$ at this price. After she offers $p\in (p_{\theta_1}-\nu,p_{\theta_1}+\nu)$, there is trade with delay less than $\eta$ conditional on the seller's cost being $\theta_1$, and there is an expected delay within $\eta$-distance from $\frac{1}{2}$ when the seller's cost is $\theta_2$. Consequently, letting $\rho^*$ denote the limiting probability with which the buyer offers $p_{\theta_2}$, type-$\theta_1$ seller's payoff in the bargaining stage converges to $\rho^*(p_{\theta_2}-\theta_1)+(1-\rho^*)(p_{\theta_1}-\theta_1)$ and type $\theta_2$'s payoff converges to $p_{\theta_2}-\theta_2$. The seller's indifference condition at the adoption stage requires that
    \begin{equation}
        \rho^*=\frac{2c -(\theta_2-\theta_1)}{\theta_2-\theta_1}.
    \end{equation}
Note that $\rho^*\in [0,1]$ if and only if $c\in \big[\frac{\theta_2-\theta_1}{2},\theta_2-\theta_1\big]$. This implies that an equilibrium with adoption probability converging to $\pi^*$ cannot be sustained if $p_{\theta_1}<\theta_2$ and $c\notin \bigg[\frac{\theta_2-\theta_1}{2},\theta_2-\theta_1\bigg]$. This conclusion will be used in the proof of Lemma \ref{lemmab5}.

Therefore,  the expected delay in the limit where $\varepsilon\to 0$ for a fixed $\nu$, and then taking the limit as $\nu\to 0$ is
    \begin{equation}
    (1-\pi^*)(1-\rho^*)\frac{1}{2}= \frac{\theta_2-\theta_1-c}{1-\theta_1}>0.
    \end{equation}
\end{proof}

Lemma \ref{lemmab3} establishes the third part of Theorem \ref{Theorem2}. In order to show the second part, we need to consider the case where (\ref{costdifference}) is satisfied but $p_{\theta_1}>\theta_2$, or equivalently $\theta_2-\theta_1\in \big(\frac{1-\theta_2}{2},1-\theta_2\big)$, and the adoption cost is intermediate. Lemma \ref{lemmab4} characterizes the limiting equilibria in this case.
\begin{Lemma} \label{lemmab4}
If $(\theta_1,\theta_2)$ satisfies (\ref{costdifference}), $p_{\theta_1}>\theta_2$, and $c\in \big(\frac{(1-\theta_2)(\theta_2-\theta_1)}{1-\theta_1},\theta_2-\theta_1\big)$, then for every $\eta>0$, there exists $\bar{\nu}>0$ such that when $\nu<\bar{\nu}$, there exists $\bar{\varepsilon}_{\nu} >0$ such that when $\varepsilon \in (0,\bar{\varepsilon}_{\nu})$:  in every equilibrium, \textit{either} the adoption probability is $\eta$-close to $\pi^*$, and the expected delay is bounded above zero;
\textit{or} the adoption probability is greater than $1-\eta$, and the expected delay is less than $\eta$. Moreover, an equilibrium of the first kind exists. 
\end{Lemma}
\begin{proof}
Fix $\eta>0$ and any small enough $\nu>0$. First, suppose by way of contradiction that there exists $\eta>0$ such that for all $\overline{\varepsilon}>0$ there exists $\varepsilon\in(0,\overline{\varepsilon})$ such that $\pi^*-\pi(\theta_1)>\eta$. Then, analogous to the proof of Lemma \ref{lemmab3}, we can show that by taking $\varepsilon$ and $\nu$ to be small enough, type-$\theta_1$ seller's equilibrium payoff can be made arbitrarily close to $p_{\theta_2}-\theta_1-c$ and type-$\theta_2$ seller's equilibrium payoff to $p_{\theta_2}-\theta_2$. Therefore, there exists $\overline{\nu}>0$ and $\overline{\varepsilon}_\nu>0$ such that $\nu<\overline{\nu}$ and $\varepsilon\in(0,\overline{\varepsilon}_\nu)$ implies that the seller strictly prefers to adopt the technology at cost   $c<\theta_2-\theta_1$. This contradicts the hypothesis that $\pi(\theta_1)<1$.

Next, suppose by way of contradiction that there exists $\overline{\varepsilon}>0$ such that $\pi(\theta_1)\in(\pi^*+\eta,1-\eta)$ when $\varepsilon\in(0,\overline{\varepsilon})$. Following the same argument as in the proof of Lemma \ref{lemmab3}, we know that the buyer will offer $\theta_2$ with probability converging to $1$ in the limit as $\varepsilon \rightarrow 0$ for $\nu>0$ fixed and then taking $\nu\to 0$. As a result, the limiting equilibrium payoff of type $\theta_1$ equals $1-\theta_2-c$, and the limiting equilibrium payoff of type $\theta_2$ equals $\frac{1-\theta_2}{1-\theta_1}(1-\theta_2)$. The assumption that $c>\frac{(1-\theta_2)(\theta_2-\theta_1)}{1-\theta_1}$ then implies that, for $\varepsilon$ and $\nu$ sufficiently small, type $\theta_2$'s payoff is strictly more than that of type $\theta_1$'s. This contradicts our earlier hypothesis that $\pi(\theta_1)>\pi^* \geq 0$. 

Therefore, the equilibrium adoption probability must converge to either $\pi^*$ or $1$. We now show that, if the limit point is $1$, then expected delay vanishes in the limit. Fix $\varepsilon\in(0,\overline{\varepsilon})$ and let $p_b$ be any price offered by the buyer with probability greater than $\eta$. In the bargaining stage, by offering $p_1(p_b)$, the low type can secure a payoff of at least $p_1(p_b)-\theta_1-\eta/2\geq p_{\theta_1}-\theta_1-\eta/2$. As a result, the buyer's bargaining payoff conditional on facing the low-cost seller is at most $1-p_1(p_b)\leq 1-p_{\theta_1}$. If $\pi(\theta_1)=1$, then the arguments in Abreu and Gul (2000) imply that the buyer can secure a payoff greater than $1-p_{\theta_1}-\eta/2$ when she offers $p_{\theta_1}$. As we show in Online Appendix A, for every $\varepsilon>0$, players payoffs in the bargaining game are continuous in $\pi(\theta_1)$. As a result, the fact that $\pi(\theta_1)$ converges to $1$ implies that the buyer's payoff when she offers $p_{\theta_1}$ is also bounded below by $1-p_{\theta_1}-\eta/2$ when $\varepsilon$ is small enough. Then, we have that, for sufficiently low $\varepsilon$, the sum of players payoffs conditional on $\theta=\theta_1$ is greater than $1-p_{\theta_1}-\eta/2+p_{\theta_1}-\theta_1-\eta/2=1-\theta_1-\eta$, which establishes the result.  

If the limit point is $\pi^*$, then it must be the case that the buyer mixes between offering prices that converge to $\theta_2$ and $p_{\theta_2}$ in a way that makes the seller indifferent
between adopting and not adopting the technology. 
As in the proof of Lemma \ref{lemmab3}, let $\rho^*$ denote the limiting probability with which the buyer offers $p_{\theta_2}$. The seller's incentive constraint at the adoption stage requires that in the limit as $\varepsilon\to 0$ for a fixed $\nu>0$, and then taking the limit as $\nu\to 0$
\begin{align*}
     \rho^* (p_{\theta_2}-\theta_1)+(1-\rho^*)(1-\theta_2)-c=\rho^*(p_{\theta_2}-\theta_2)+(1-\rho^*)\frac{1-\theta_2}{1-\theta_1}(1-\theta_2)  
\end{align*}
or equivalently,
\begin{equation}
\rho^*=\frac{(1-\theta_1)c-(1-\theta_2)(\theta_2-\theta_1)}{(\theta_2-\theta_1)^2}.
\end{equation}
With $\rho^*\in [0,1]$, we have $c\in\bigg[\frac{(1-\theta_2)(\theta_2-\theta_1)}{1-\theta_1},\theta_2-\theta_1\bigg]$.
The expected delay, in the limit as $\varepsilon\to 0$ for a fixed $\nu$, and then taking the limit as $\nu\to 0$, is then
\begin{equation}
    (1-\pi^*)(1-\rho^*)\Bigg(1-\frac{1-\theta_2}{1-\theta_1}\Bigg)=\frac{(3\theta_2-1-2\theta_1)(\theta_2-\theta_1-c)}{2(\theta_2-\theta_1)^2}>0.
\end{equation}

The above conditions must be satisfied in any equilibrium in which the seller's adoption probability is arbitrarily close to $\pi^*$. Next, we show that such an equilibrium exists. Let $\pi(\varepsilon,\nu)\in(0,1)$ (which is arbitrarily close to $\pi^*$) denote the adoption rate that makes the buyer exactly indifferent between offering $p_b\in p_1^\nu=\argmax\limits_{p\in \mathbf{P}_b\cap[\theta_1,\theta_2]}(1-\max\{p,p_1(p)\})$ and $p_b\in p_2^\nu=\argmax\limits_{p\in \mathbf{P}_b\cap[\theta_2,1]}(1-\max\{p,p_2(p)\})$. As we will show in Online Appendix A, the continuation war-of-attrition game with $\pi(\theta_1)=\pi(\varepsilon,\nu)$ has an equilibrium. Thus, it suffices to show that given players' equilibrium strategies in the war-of-attrition game, the seller is indifferent at the adoption stage and hence choosing $\pi(\varepsilon,\nu)\in (0,1)$ is incentive compatible. 

As we have shown earlier, when $\varepsilon$ and $\nu$ are arbitrarily small, players' on-path strategies in the war-of-attrition game are arbitrarily close to the following: The buyer offers $p_{\theta_2}$ with probability $\rho^*\in [0,1]$ and offers $\theta_2$ with probability $1-\rho^*$. After the buyer offers $p_{\theta_2}$, both types of the seller accept the offer. After the buyer offers $\theta_2$, type $\theta_1$ demands $1+\theta_1-\theta_2$ and type $\theta_2$ demands $1$. Moreover, when $\varepsilon>0$ and $\nu>0$ are sufficiently small, we know that if $\rho^*=1$, then the seller's equilibrium payoff from adopting the technology is strictly greater than his payoff from not adopting, while the opposite is true if $\rho^*=0$. By continuity, there exists $\rho(\varepsilon,\nu)\in(0,1)$ such that the seller is indifferent between adopting and not adopting the technology, and therefore the interior adoption probability $\pi(\varepsilon,\nu)$ can be sustained as the seller's adoption strategy in an equilibrium of the game. 

\end{proof}

Lemma \ref{lemmab4} implies that first, as stated in Theorem \ref{Theorem2}, when (\ref{costdifference}) holds and $\theta_2<p_{\theta_1}$, there is an open set of adoption costs, given by $\Big(\frac{(1-\theta_2)(\theta_2-\theta_1)}{1-\theta_1},\theta_2-\theta_1\Big)\subset \Big(\frac{\theta_2-\theta_1}{2},\theta_2-\theta_1\Big)$ such that there exists an equilibrium with inefficient adoption and significant delay in reaching agreement. This equilibrium shares the same features as the unique equilibrium characterized in Lemma \ref{lemmab3}. Second, there may be multiple limiting equilibria.\footnote{We show that an inefficient equilibrium always exists under the conditions stated in Lemma \ref{lemmab4}. Additionally, our arguments do not allow to rule out the existence of an approximately efficient equilibrium, although we do not have an explicit construction of this equilibrium and thus our proof does not establish that this equilibrium exists.} When $\theta_2<p_{\theta_1}$ and the cost of adoption satisfies $c\in\Big(\frac{(1-\theta_2)(\theta_2-\theta_1)}{1-\theta_1},\theta_2-\theta_1\Big)$, we cannot rule out the approximately efficient equilibrium since whenever the buyer expects the seller to adopt with probability close to $1$, her optimal strategy is to offer $p_{\theta_1}$. Since $\theta_2<p_{\theta_1}$, it must be the case that $\tau_s <+\infty$ if the seller does not adopt. Therefore, conditional on not adopting, the seller will concede in finite time, and the expected delay can make him indifferent between adopting and not adopting. The fact that $\pi(\theta_1)$ is close to $1$ ensures that this delay can be sustained as an outcome of the war-of-attrition game. However, this reasoning does not apply when $\theta_2>p_{\theta_1}$. This is because conditional on not adopting, type $\theta_2$ has no incentive to concede when the buyer offers $p_{\theta_1}$. This ability to commit makes type $\theta_2$'s bargaining payoffs considerably higher, which would give rise to a profitable deviation for the seller at the adoption stage. 

Finally, we consider the remaining case in which either (\ref{costdifference}) is violated and $c<\theta_2-\theta_1$, or \eqref{costdifference} is satisfied and $c<\max\big\{\frac{1}{2},\frac{1-\theta_2}{1-\theta_1}\big\}(\theta_2-\theta_1)$. We show that in every equilibrium, the seller's investment decision is socially efficient and there is almost no delay in reaching agreement. 
\begin{Lemma} \label{lemmab5}
    If the parameters of the model satisfy (\ref{costdifference}) and $c<\max\big\{\frac{1}{2},\frac{1-\theta_2}{1-\theta_1}\big\}(\theta_2-\theta_1)$, or if (\ref{costdifference}) is violated and $c<\theta_2-\theta_1$, then for every $\eta>0$, there exists $\bar{\nu}>0$ such that when $\nu<\bar{\nu}$, there exists $\bar{\varepsilon}_{\nu} >0$ such that for every $\varepsilon \in (0,\bar{\varepsilon}_{\nu})$, the adoption probability is at least $1-\eta$ and the expected welfare loss from delay is no more than $\eta$ in every equilibrium.
\end{Lemma}

\begin{proof}
    Fix $\eta>0$ and a small enough $\nu>0$. Suppose first that (\ref{costdifference}) is violated and $c<\theta_2-\theta_1$. If for any small enough $\varepsilon>0$, the seller will adopt the technology with probability less than $1-\eta$, then the buyer's incentive constraint implies that her equilibrium offer is $\eta$-close to the Dirac measure on $p_{\theta_2}$. Therefore, type $\theta$'s payoff in the war-of-attrition game is $\eta$-close to $p_{\theta_2}-\theta$. The assumption that $c<\theta_2-\theta_1$ then implies that the seller strictly prefers to adopt. This contradicts our earlier hypothesis that $\pi(\theta_1)<1$. 
    
    Next, suppose that (\ref{costdifference}) is satisfied and $c<\max\big\{\frac{1}{2},\frac{1-\theta_2}{1-\theta_1}\big\}(\theta_2-\theta_1)$. In the subcase where there exists an arbitrarily small $\varepsilon$ such that $\pi(\theta_1)<\pi^*-\eta$, the same argument as in the previous paragraph applies. Moreover, our proofs of Lemmas \ref{lemmab3} and \ref{lemmab4} imply that an inefficient equilibrium with limiting adoption probability equal to $\pi^*$ exists only if $c>\max\big\{\frac{1}{2},\frac{1-\theta_2}{1-\theta_1}\big\}(\theta_2-\theta_1)$, which is ruled out under the assumption in Lemma \ref{lemmab5}. In the subcase where $\pi(\theta_1)\in (\pi^*+\eta,1-\eta)$ for all arbitrarily small $\varepsilon>0$, the buyer's incentive constraint requires that her offer in equilibrium belong to $(\min\{\theta_2,p_{\theta_1}\}-\nu,\min\{\theta_2,p_{\theta_1}\}+\nu)$ with probability more than $1-\eta$. As a result, the seller's payoff from adopting the technology converges to $\max\{p_{\theta_1},1+\theta_1-\theta_2\}-\theta_1-c$ and his payoff from not adopting the technology converges to $\max\big\{\frac{1}{2},\frac{1-\theta_2}{1-\theta_1}\big\}(1-\theta_2)$. Hence, the limiting payoff gain from adoption is
\begin{equation*}
    \max\Bigg\{\frac{1}{2},\frac{1-\theta_2}{1-\theta_1}\Bigg\} \Big(\theta_2-\theta_1 \Big)-c>0,
\end{equation*}
which follows from the assumption that $c<\max\big\{\frac{1}{2},\frac{1-\theta_2}{1-\theta_1}\big\}(\theta_2-\theta_1)$. This contradicts our earlier hypothesis that $\pi(\theta_1)<1-\eta$ when $\varepsilon$ is sufficiently small. 
Hence, the seller's equilibrium adoption probability must converge to $1$ in the limit as $\varepsilon \rightarrow 0$ for a fixed $\nu>0$ and then taking $\nu\to 0$. Finally, an argument analogous to the one given in Lemma \ref{lemmab4} establishes that the outcome conditional on the seller being type $\theta_1$ is approximately efficient, both in the adoption stage and in the ensuing war-of-attrition game. If this is the case, then the expected delay from reaching an agreement vanishes in the limit as $\varepsilon \rightarrow 0$. 
\end{proof}

\end{document}